\documentclass[english,reqno,dvips,12pt]{article}

\usepackage{epsfig,graphicx,amsmath,amssymb,amsfonts,amsthm,amscd}
\newtheorem{theorem}{Theorem}[section]
\newtheorem{lemma}{Lemma}

\theoremstyle{definition}

\theoremstyle{remark}
\newtheorem{remark}{Remark}

\numberwithin{equation}{section}

\newcommand{\cN}{{\mathcal N}}
\newcommand{\cM}{{\mathcal M}}
\newcommand{\cP}{{\mathcal P}}

\newcommand{\cE}{{\mathcal E}}

\newcommand{\wh}{\widehat}

\newcommand{\pa}{\partial}

\newcommand{\nn}{\nonumber}

\begin{document}

\title{Mathematical conception of  the gas theory}
\author{V.P.Maslov}


\date{}
\maketitle

\begin{abstract}
In this paper, using of the rigorous statement and rigorous proof
the Maxwell distribution as an example, we establish estimates of
the distribution depending on the parameter~$N$, the number of
particles. Further, we consider the problem of the occurrence of
dimers in a classical gas as an analog of Bose condensation and
establish estimates of the lower level of the analog of Bose
condensation. Using of the dequantization principles we find the
relationship of this level to ``capture'' theory in the
scattering problem corresponding to an interaction of the form of
the Lennard-Jones potential. This also solves the problem of the
Gibbs paradox.

We derive the equation of state for a nonideal gas as a result of
pair interactions of particles in Lennard-Jones models and, for
classical gases, discuss the $\lambda$-transition to the
condensed state (the state in which $V_{\text{sp}}$ does not vary
with increasing pressure; for heat capacity, this is the
$\lambda$-point).
\end{abstract}

\section{The Maxwell distribution}

It is an old misconception that statistical physics and thermodynamics
can be derived from the laws of mechanics and dynamical systems.
It still persists from the days of the controversy between
L.~Boltzmann with H.~Poincar\'e, E.~Zermelo, and other mathematicians.
However, in order to solve the clusterization problem,
computer simulation based, as a rule, on the laws of mechanics is
usually used.  For example, from a mechanical point of view, to
obtain a dimer, i.e., a coupled pair of particles, a
``collision'' (interaction) of three particles must occur (so
that part of the energy is imparted to one of them, and only this
results in a mechanical capture).  Although modern computers are
powerful, one can \textit{a priori} expect the wrong answer.
The main equation of statistical physics, the Boltzmann equation,
cannot be obtained, in principle, solely from mechanical laws.
In principle, it can be derived (but not in the near future)
from quantum field theory, in which there are no \textit{a
priori} prescribed pair interactions.  Here we rely on the
natural axiom of the existence of a mean field formed by $N$.
particles, probabilistic number theory, and the theory of white
noise at a given temperature~\cite{31:v851}.

Academcian N.N.Bogolyubov used to say: ''I looked fora small
parameter during my whole life.'' Bogolyubov was a mathematician
in essence and looked for small parameters in physics. Physicists
operate very well  with digital values and intuitively understand
or  reckon up mentally whether or not a given asymptotic is
applicable.

Let us cite the corresponding text of Landau and Lifshits in
their manual on statistical physics. Assuming that the
Russel-Sounders case of connection in the atom holds, the authors
represent the partition function in the following form (we
simplify their representation):
\begin{equation}\label{RJ4}
Z= \sum e^{-\varepsilon_j/kT},
\end{equation}
where the symbols $\varepsilon_j$ stand for the components of the
fine structure of the normal term. Let us quote: ''As is known,
the existence of nuclear spin leads to the so-called hyperfine
splitting of atomic levels. However, the intervals of this
structure are so tiny that they can be regarded as small
intervals as compared with $T$ for all the temperatures for which
the gas exists as gas.'' This is continued in a footnote: ''The
temperatures corresponding to intervals of the hyperfine
structure of diverse atoms are beyond the limits from $0,1^0$ to
$1,5^0$." (\cite{Landau}, Russian p. 163). Thus, the authors say
that the value $T = 0,1K$ is large, which just means that one must
introduce the small parameter indicated above.

This example shows that physicists do not need this parameter.
Using at appropriate places the digits of hyperfine structure,
the Russel-Sounders connection, and the digits arising in the
specific problem under consideration, they know whether or not a
given formula is applicable better than the mathematician who
obtained the related estimates. However, to be correct, this is
true for the Great Physicists only, and the above instruments can
create far-reaching errors of ordinary good physicists (see,
e.g.,\cite{JETF}).

Let us consider a classical gas.

The Maxwell distribution is of the form
\begin{equation}
\label{eq1:v851} \cN_{dv} =\biggl(\frac{m}{2\pi kT}\biggr)^{3/2}
e^{-(m(v^2_x+v^2_y+v^2_z))/2kT}\,dv_x\,dv_y\,dv_z,
\end{equation}
where $v$ is the velocity, $T$ is the temperature, $k$ is the
Boltzmann constant, and $\cN_{dv}$ is the relative number of
particles contained in the interval $dv=dv_x\,dv_y\,dv_z$.

We must bear in mind that its interpretation as a distribution
density is false in the general case, but valid only in the
cumulative variant. This implies that the integral of the
density~\eqref{eq1:v851} over any narrow finite velocity interval
bounded below determines, indeed, the relative
number~$\cN_{v_1v_2}$ of particles in this velocity range.

We obtain the usual Maxwell distribution for a sufficiently
narrow velocity interval:
\begin{equation}
\label{eq2:v851} \cN_{v_1\sqrt{\kappa},v_2\sqrt{\kappa}}
=\frac{\int_{|v_1|\sqrt{\kappa}}^{|v_2|\sqrt{\kappa}} e^{-
mv^2/2kT}v^2\,d|v|} {\int_{-\infty}^\infty
e^{-mv^2/2kT}v^2\,d|v|}\,,
\end{equation}
where $\cN_{v_1\sqrt\kappa,v_2\sqrt\kappa}$ is the relative
number of particles with velocities in the interval
$|v_1|\sqrt\kappa$, $|v_2|\sqrt\kappa$, and $\kappa$ is a small
parameter.

By a change of variables, we can get rid of the small
parameter~$\kappa$ by transferring it to the exponential, i.e.,
\begin{equation}
\label{eq3:v851} \cN_{v_1v_2} =\frac{\int_{|v_1|}^{|v_2|}
e^{-(\kappa mv^2)/2kT}v^2\,d|v|} {\int_{-\infty}^\infty
e^{-(\kappa mv^2)/2kT}v^2\,d|v|}\,, \qquad \text{for}\quad
\kappa\to0.
\end{equation}

This formula, as an asymptotic formula, will be obtained below,
as well as its estimate, i.e., its domain of applicability.

Consider the most often used Lennard-Jones interaction potential
\begin{equation}
\label{eq4:v851} \varphi(r) =4U_0\biggl[\biggl(\frac\sigma
r\biggr)^{12} -\biggl(\frac\sigma r\biggr)^{6}\biggr],
\end{equation}
where $\sigma$ is the distance at which the potential function
changes sign and $U_0$ is the minimum value of the potential (at
the point $r=2^{1/6}\sigma$) or the depth of the potential well.

From dimensional considerations~\eqref{eq1:v851} for the
quantities appearing in the definition of the particle (neutral
molecule), we can write
\begin{equation}
\label{eq5:v851} \biggl(\frac{m}{U_0}\biggr)^{3/2}
\int^\infty_{-\infty} e^{-(\kappa mv^2)/2kT}\,dv_x\,dv_y\,dv_z=N.
\end{equation}
for this molecule. Hence
\begin{equation}
\label{eq6:v851} \frac43\cdot\frac{\pi}{\kappa^{3/2}}
\biggl(\frac{2kT}{U_0}\biggr)^{3/2}=N,
\end{equation}
i.e.,
$$
\kappa=\biggl(\frac{2kT}{U_0}\biggr)N^{-2/3}
\cdot\biggl(\frac{4\pi}3\biggr)^{2/3}.
$$
Therefore, the velocity interval on which we can determine the
relative number of particles is
\begin{equation}
\label{eq7:v851} \biggl\{\biggl(\frac{2kT}{U_0}\biggr)^{1/2}
|v_1|N^{-1/3}, \biggl(\frac{2kT}{U_0}\biggr)^{1/2}|v_2|
 N^{-1/3}\biggr\},
\end{equation}
where $v_1$ and $v_2$ arbitrary velocities, $|v_2|>|v_1|$,
independent of the number~$N$.

For a rigorous justification of the Maxwell distribution, we use
the seemingly insignificant fact that the number of particles~$N$
is an integer and apply number theory, which does not seem
relevant at all.

The Maxwell distribution is equally important in complexity
theory as the Poisson, Gauss, and other classical distributions.

Let us define the energy
\begin{equation}
\label{eq8:v851}
 E=\frac{4\pi/3\cdot(kT/2)^{5/2}}{U_0^{3/2}}\,.
\end{equation}
The distribution in the energy interval between $mv_1^2/2$ and
$mv_2^2/2$ must have the following form:
\begin{equation}
\frac{E_{v_1v_2}}{E} =\frac{\int_{mv_1^2/2}^{mv_2^2/2}\xi
e^{-\kappa(\xi/kT)} \,d\xi^{3/2}} {\int_{0}^{\infty}\xi
e^{-\kappa(\xi/kT)}\,d\xi^{3/2}}
=\frac{\int_{\kappa(mv_1^2/2)}^{\kappa(mv_2^2/2)} \xi
e^{-(\xi/kT)}\,d\xi^{3/2}} {\int_{0}^{\infty}\xi
e^{-(\xi/kT)}\,d\xi^{3/2}}\,. \label{eq9:v851}
\end{equation}
Without loss of generality, let
$$
\frac{mv_1^2}2=lU_0, \qquad \frac{mv_2^2}2=(l+1)U_0,
$$
where $l$ is an integer. Hence, up to $O(\kappa^2)$, we can write
\begin{equation}
\label{eq10:v851} \frac{E_{v_1v_2}}{E} =\frac{U_0\kappa
l^{1/2}l}{E} +O(\kappa^2)l^{1/2}\mspace{1mu}\frac{U_0}{E}\,.
\end{equation}
If we split the total energy~$E$ into intervals of the form
$\{l{U_0},(l+1)U_0\}$, $l=0,1,\dots,l_E$, so that the sum of
these intervals is less than~$E$ by a quantity $O(\kappa)$, then,
in view of the Euler--Maclaurin formula, we have the order
$$
l_E\cong\kappa^{-7/5}.
$$
Further, if we replace~$l^{1/2}$ by its integer part~$[l^{1/2}]$,
then we decrease the sum of the intervals by at most a quantity
$O(\kappa^{-7/5})$.

Hence the union of the partitions
$\kappa\sum_0^{l_E}U_0[l^{1/2}]l$ satisfies the inequalities
\begin{equation}
\label{eq11:v851}
 E-O(\kappa^{-7/5})
\le\kappa U_0\sum_0^{l_E}[l^{1/2}]l\le E.
\end{equation}
Thus, we obtain energy boxes and wish to find the most probable
number of particles with energies in each box.

Now let us split the number of particles $N=\sum N_{jk}$, where
$k=1,2,\dots,[j^{1/2}]$, $j=1,2,\dots,l_E$.

Therefore, given condition~\eqref{eq11:v851}, we obtain the
following constraint on our partition:
\begin{equation}
\label{eq12:v851}
 E-O(\kappa^{-7/5})
\le\kappa U_0\sum_{j=1}^{l_E} \sum_{k=1}^{[j^{1/2}]}N_{jk}\le E.
\end{equation}

The condition $N=\sum N_{jk}$ implies that the size of the
ordered sample with replacement~\cite{1:v851} is equal to~$N$,
while condition~\eqref{eq12:v851} means that the energy
corresponding to this sample is contained in the interval
$\{E-O(\kappa^{-7/5}),E\}$.

These were heuristic considerations. Now we make the following
assumptions.

In the volume~$V$, consider the system of $N$ particles
possessing the energy~$E$. Moreover, $N\to\infty$,
$E/U_0\to\infty$.

The interval $(O,E)$ is divided into small (compared with~$E$)
subintervals $E_i<E_{i+1}$, $i=1,\dots,l_0$, and the
corresponding intervals of the moduli of velocities
$|v_i|<v_{i+1}$, as well as intervals of the phase volume, and
the energy boxes $\Delta\Omega_i$ that are contained between
these velocities
\begin{equation}
\label{eq13:v851} \Delta\Omega_i=\int_{|v_i|}^{|v_{i+1}|}
\frac{mv^2}{2}\,dv_1dv_2dv_3
=\text{const}(E_{i+1}^{5/2}-E_{i}^{5/2})V
\approx\text{const}\Delta E_iE_i^{3/2}V.
\end{equation}
In these energy boxes, we place different particles using all
possible ways~$N$. In other words, we take an ordered sample with
replacement from~$N$ ``balls'' to these energy boxes (phase
volumes) by the method indicated in \eqref{eq12:v851}:
\begin{equation}
\label{eq14:v851} 0\le\kappa U_0\sum_{j=1}^{l_E}
\sum_{k=1}^{[j^{1/2}]}N_{jk}\le E.
\end{equation}

By $\cN(|v_1|,|v_2|)$ we denote the relative number of particles
in the velocity interval $(v_1,v_2)$.

Under the conditions given above, the following theorem is valid.

\begin{theorem}
\label{t1:v851} The probability that the estimate
\begin{align}
\nonumber &\cN(|v_1|,|v_1|+O(N^{-1/2+\delta})) -\frac{
\int_{|v_1|}^{|v_1|+O(N^{-1/2+\delta})}
e^{-mv^2/2kT}\,dv_1\,dv_2\,dv_3} {\int_{-\infty}^{-\infty}
e^{-mv^2/2kT}\,dv_1\,dv_2\,dv_3}
\\&\qquad
=O\biggl(\frac{\sqrt{\ln N}}{\sqrt{N}} |{\ln\ln
N}|^\varepsilon\biggr), \label{eq15:v851}
\end{align}
where $\delta>0$, and $\varepsilon$ is any arbitrarily small
number, and $|v_1|\ge0$ are arbitrary velocities), does not hold
is exponentially small (is less than~$1/N^k$, where $k$ is any
integer).
\end{theorem}

The same assertion is also valid for any large velocity interval,
i.e., $v_2>\delta>0$, where $\delta$ is independent of~$N$.

By analogy with the term ``convergence in measure,'' we can state
that, in \eqref{eq15:v851}, there is an ``estimate in measure.''
In fact, the physical formula~\eqref{eq1:v851} can be rewritten
in the more exact form
\begin{equation}
\label{eq16:v851} \cN_{\Delta v}
=\int_{v_1}^{v_1+O(N^{-1/2+\delta})} \biggl(\frac{m}{2\pi
kT}\biggr)^{3/2}
e^{-(m(v^2_x+v^2_y+v^2_z))/2kT}\,dv_x\,dv_y\,dv_z,
\end{equation}
where $\Delta=O(N^{-1/2+\delta})$ and $\delta>0$.

In usual probability notation, the theorem can be restated as
follows.

\begin{theorem}
\label{t2:v851} The following relation holds:
\begin{align}
&\mathsf{P}\Biggl(\cN(|v_1|,|v_1|+O(N^{(-1/2)+\delta}))
-\frac{\int_{|v_1|}^{|v_1|+O(N^{(-1/2)+\delta})}
e^{-mv^2/2kT}\,dv_1\,dv_2\,dv_3}
{\int_{-\infty}^{-\infty}e^{-mv^2/2kT}\,dv_1\,dv_2\,dv_3}
\ge\frac{\sqrt{\ln N}}{\sqrt{N}} |{\ln\ln N}|^\varepsilon\Biggr)
\nonumber \\&\qquad =O(N^{-k}), \label{eq17:v851}
\end{align}
where $k$ is any number, $\varepsilon>0$ is an arbitrarily small
number, $\delta$ is any number, and $|v_1|\ge\nobreak0$ are
arbitrary velocities. Here $\mathsf P$ is the Lebesgue measure of
the phase volume defined in parentheses in \eqref{eq17:v851} with
respect to the total volume.
\end{theorem}

These estimates are sharp (unimprovable). The theorem belongs
to~number theory. It has no relation to particle dynamics in
which the Maxwell distribution is derived from the Boltzmann
equation, which has not been is rigorously justified up to now.
The usual dynamical approach and its criticism is contained in
Kozlov's book~\cite{2:v851}.

Nevertheless, it is natural that, under certain conditions. the
dynamical system attains the most probable (from the point of
view of probability number theory) distribution. This
consideration can be useful for the dynamical approach.

Let us present a sufficiently elementary proof the theorem on the
Maxwell distribution, without, essentially, referring to
important and elegant results of number theory based on the
Meinardus theorem~\cite{3:v851}, Theorem 6.2  and on Vershik's
elegant theory of multiplicative measures~\cite{4:v851}, which
could help us avoid some inessential and deliberate
manifestations of integrality (such as taking the integer part
of~$l^{1/2}$).

On the other hand, the given estimates, which the author used in
his papers dealing with economics and
linguistics~\cite{5:v851},~\cite{6:v851}, are more understandable
to readers that are not experts in number theory and probability
theory, in particular, to physicists and specialists in analysis.

The proof is based on the estimates given by the author
in~\cite{7:v851} and on a theorem similar to the Meinardus
theorem.

For the $6$-dimensional case, a detailed proof was given
in~\cite{8:v851}. Essentially, we repeat this proof for the
$3$-dimensional case.

Let us study the system defined as follows. For energy levels
$j=0,1,2,\dots$ of multiplicities
\begin{equation}
\label{eq18:v851} q_i=[j^{1/2}],\qquad j=0,1,2,\dots,
\end{equation}
we consider all possible collections $\{N_{jk}\}$ of nonnegative
integers~$N_{jk}$, $j=0,1,2,\dots$, $k=1,\dots,q_i$ satisfying
the conditions
\begin{gather}
\label{eq19:v851} \sum_{j=0}^\infty N_{jk}=N ;
\\
\label{eq20:v851} \sum_{j=0}^\infty\sum_{k=1}^{q_i}jN_{jk}
\equiv\sum_{j=1}^\infty\sum_{k=1}^{q_i}j N_{jk}\le M,
\end{gather}
where $N$ and $M$ are given positive numbers (which can be
assumed integers without loss of generality) All such collections
are assumed equiprobable.

Denote the dimensionless quantity $M=E/U_0$ and define the
numbers~$\beta$ and $\xi$ as solutions of the system of equations
$$
M=\xi^{-1}\sum_{j=1}^\infty j[j^{1/2}]e^{-\beta j},\qquad
N=\xi^{-1}\sum_{j=0}^\infty[j^{1/2}]e^{-\beta j}.
$$

\begin{remark}
By the Euler--Maclaurin formula, we have
\begin{align*}
 N&=\xi^{-1}\int_0^\infty x^{1/2}e^{-\beta x}
\,dx\,(1+O(\beta)) \cong\frac12\mspace{1mu}\sqrt\pi\mspace{2mu}
\beta^{-3/2}\xi^{-1},
\\
 M&=\xi^{-1}\int_0^\infty x^{3/2}e^{-\beta x}
\,dx\,(1+O(\beta)) \cong\frac34\mspace{1mu}\sqrt\pi\mspace{2mu}
\beta^{-5/2}\xi^{-1}.
\end{align*}

Hence $\beta$ and $\xi$,
$$
\beta\cong\frac23\frac NM\,,\qquad \xi\cong\frac12\mspace{1mu}
\sqrt\pi\mspace{2mu}\frac{M^{3/2}}{N^{5/3}}\,.
$$
\end{remark}

\begin{remark}
As an example, we consider the Maxwell distribution.
Note that the same argument can be used for any arbitrary Gibbs distribution,
but with more cumbersome estimates.
Just the same estimates are obtained for the Gibbs distribution
corresponding to the Hamiltonian $E=(p^2+q^2)^2$,
because the phase cells for this Hamiltonian
satisfy the same relations \eqref{eq6:v851}--\eqref{eq14:v851}.

We consider sufficiently general classical Hamiltonian function
$H(p,q)$, where $q\in R^3$,  $p\in R^3$,
under the following two assumptions:

1) $H(p,q) \to \infty$  as $|p|+|q| \to \infty$ not slower than
$(|p|+|q|)^\alpha$ for some $\alpha>0$;

2) the function
$$
V(\Lambda) = \int_{H(p,q)\leq\Lambda} dpdq =\int_0^\Lambda dE
\iint \delta\left(E-H(p,q)\right) dpdq,
$$
under the assumption that $V'(\Lambda)$ is a sufficiently smooth function,
determines a phase cell
invariant under the Hamiltonian system corresponding
to the Hamiltonian $H(p,q)$.

We choose a partition such that $E_{l+1}-E_l = E_0$. Then
\begin{equation}\label{ad1}
E_l=E_0 (l+1).
\end{equation}
Let~$N_l$ be an ordered sample with replacement to the cell
$E_{l+1}-E_l$.
An \textit{ordered sample} with replacement from~$N$ balls
to cells invariant under the Hamiltonian system
(to ''energy boxes'')
$$
\int_{E_l}^{E_{l+1}} H(p,q) dp dq = \int_{E_0l}^{E_0(l+1)}\lambda
V'(\lambda) d \lambda
$$
\textit{leads to the state}
\begin{equation}\label{ad2}
\sum_l N_lE_lq_l \leq \cE_N, \quad q_l\cong  C[V'(E_0 l)],
\end{equation}
where $C$ is a constant.
Then the proof and the estimates are just the same
as in Theorem~\ref{t2:v851}.
\end{remark}

Denote by $\cN(M,N)$ the total number of collections $\{N_{jk}\}$
satisfying the constraints~\eqref{eq19:v851}~\eqref{eq20:v851}.

We assume everywhere that the parameters~$\beta$ and $\xi$
satisfy the relation
\begin{equation}
\label{eq21:v851} \xi<\beta^{-3/2+\varepsilon},
\end{equation}
for an arbitrary (but fixed) $\varepsilon>0$.

Let
$$
\mu=\ln\xi,\qquad \text{so that}\quad \xi=e^{\mu}.
$$

Suppose that $\cM$ is the set of ordered samples satisfying
conditions~\eqref{eq19:v851} and \eqref{eq20:v851}.

For the numbers $\cN(N,M)$ of such variants, we obtain the
following estimate:
\begin{equation}
\label{eq22:v851} \cN(M,N) \le C\sqrt N\exp\{N\ln N+\beta M+\mu
N\}.
\end{equation}
 Indeed,
\begin{align*}
\cN(M,N) &=N!\sum_{\{N_{jk}\}\in\cM} \frac1{
\prod_{j=0}^\infty\prod_{k=1}^{q_j}N_{jk}!}
\le N!\mspace{2mu} e^{\beta M+\mu N}\sum_{\{N_{jk}\}}
\frac{\exp\bigl\{-\sum_{j=0}^{\infty}
\sum_{k=1}^{q_j}N_{jk}(\beta j+\mu)\bigr\}} {
\prod_{j=0}^\infty\prod_{k=1}^{q_j}N_{jk}!}
\\&
=N!\mspace{2mu}e^{\beta M+\mu N}
\prod_{j=0}^\infty\prod_{k=1}^{q_j} \sum_{N_{jk}=0}^\infty
\frac{e^{-N_{jk}(\beta j+\mu)}}{N_{jk}!}
=N!\mspace{2mu}e^{\beta M+\mu N}
\prod_{j=0}^\infty\prod_{k=1}^{q_j} e^{e^{-\beta j-\mu}}
\\&
=N!\exp\biggl\{\beta M+\mu N +\sum_{j=0}^{\infty}q_je^{-\beta
j-\mu}\biggr\}
=N!\exp\{\beta M+\mu N+N\}
\\[1mm]&
\le C\sqrt N\exp\{N\ln N+\beta M+\mu N\} \qquad\qquad\qquad\quad
\text{(by Stirling's formula)}.
\end{align*}

Suppose that $\cM_\Delta\subset\cM$ is the subset of variants
such that
\begin{equation}
\label{eq23:v851} \biggl|\sum_{j=0}^{l}\sum_{k=1}^{q_j}
(N_{jk}-\overline N_{jk})\biggr|>\Delta,
\end{equation}
where
\begin{equation}
\label{eq24:v851} \overline N_{jk} =e^{-\beta
j-\mu}\equiv\frac{\overline N_j}{q_j}\,.
\end{equation}
For the number $\cN(M,N,\Delta)$ of the sample from~$\cM_\Delta$,
we obtain the estimate {\begin{align} \nonumber \cN(M,N,\Delta)
&\le N!\exp\biggl\{\beta M+\mu N-c\Delta +\sum_{j=l+1}^\infty
q_je^{-\beta j-\mu}\biggr\} \nonumber
\\&\qquad \times\biggl(\exp\biggl\{\sum_{j=0}^{l} (q_je^{-\beta
j-\mu+c}-c\overline{N}_j)\biggr\} +\exp\biggl\{\sum_{j=0}^{l}
(q_je^{-\beta j-\mu-c}+c\overline{N}_j)\biggr\}\biggr)
\label{eq25:v851}
\end{align}}
for $0<c<\mu$, where $\overline N_j$ is given
by~\eqref{eq24:v851}.

Further, as in~\cite{NonlinearAv}, we take two terms of the
expansion in the Taylor series
\begin{equation}
\label{eq26:v851} q_je^{-\beta j-\mu\pm c}\mp c\overline{N}_j
=q_je^{-\beta j-\mu}(e^{\pm c}\mp c) =q_je^{-\beta j-\mu}
\biggl(1+\frac{c^2}2\mspace{1mu}e^{\pm\theta c}\biggr),
\end{equation}
where $\theta(i)$ is some midpoint,
$\theta\equiv\theta(c)\in(0,1)$.

If $c\le\min\{\mu/2,1\}$, then this implies the inequality
\begin{equation}
\label{eq27:v851} \sum_{j=0}^{l}q_je^{\beta j+\mu-\theta_j c}
\le2Ke^{-\mu}\beta^{-3/2},
\end{equation}
where $K$ is a constant.

Therefore,
\begin{equation}
\label{eq28:v851} \cN(N,M,\Delta) \le C\sqrt N\exp\{N\ln N+\beta
M+\mu N\} \exp\{-c\Delta+Kc^2e^{-\mu}\beta^{-3/2}\}.
\end{equation}
 We substitute
\begin{equation}
\label{eq29:v851} \Delta=\sqrt{N\ln N}\mspace{2mu} |{\ln\ln
N}|^\varepsilon \asymp e^{-\mu/2}\beta^{-3/2} \sqrt{\ln
N}\mspace{2mu}|{\ln\ln N}|^\varepsilon
\end{equation}
and
\begin{equation}
\label{eq30:v851} c=\frac{\beta^{3/2}e^\mu\Delta}{2K}
\asymp\beta^{3/4}e^{\mu/2} \sqrt{\ln N}\mspace{2mu}|{\ln\ln
N}|^\varepsilon
\end{equation}
in \eqref{eq28:v851}.
 This implies that,
in particular,
\begin{equation}
\label{eq31:v851} \cN(N,M,\Delta) \le C_k\sqrt N\exp\{N\ln
N+\beta M+\mu N\}N^{-k}
\end{equation}
for any~$k$.

Let us now find a lower bound for these quantities.

We estimate the number of samples $\cN_0(M,N)<\cN(M,N)$
satisfying conditions~\eqref{eq19:v851} and \eqref{eq20:v851};
moreover, in the last inequality, we consider the equality
\begin{equation}
\label{eq32:v851}
\sum_{j=0}^{\infty}\sum_{k=1}^{q_j}N_{jk}=N,\qquad
\sum_{j=0}^{\infty}\sum_{k=1}^{q_j}jN_{jk}=M.
\end{equation}
 Suppose that~$\cM_0$
is the set of collections of occupation numbers
satisfying~\eqref{eq32:v851}. Then
\begin{align}
\nonumber \cN_0(M,N) &=\sum_{\{N_{jk}\}\in\cM_0} \frac{N!}{
\prod_{j=0}^\infty\prod_{k=1}^{q_j}N_{jk}!}
\\&
=N!\sum_{\{N_{jk}\}} \frac{\delta\bigl(N,\sum_{j=0}^{\infty}
\sum_{k=1}^{q_j}N_{jk}\bigr)
\delta\bigl(M,\sum_{j=0}^{\infty}\sum_{k=1}^{q_j} jN_{jk}\bigr)}
{\prod_{j=0}^\infty\prod_{k=1}^{q_j}N_{jk}!}\,. \label{eq33:v851}
\end{align}
Here the sum in the second row is taken over all finite
collections of nonnegative occupation numbers and
$\delta(m,n)\equiv\delta_{mn}$ is the Kronecker delta.

Substitute the integral representation
$$
\delta_{mn}=\frac s{2\pi}\int_{-\pi/s}^{\pi/s}
e^{(isx+\omega)(m-n)}\,dx
$$
of the Kronecker symbol (where~$s$ and $\omega$ are arbitrary
nonzero real numbers) into~\eqref{eq33:v851}, choosing $s=1$ and
$\omega=\mu$ for the first factor and $s=\omega=\beta$ for the
second factor. Then, for $\cN_0(M,N)$, we obtain the integral
representation
\begin{equation}
\label{eq34:v851} \cN_0(M,N) =\frac{\beta N!\mspace{2mu}e^{\beta
M+\mu N}}{4\pi^2} \int_{-\pi/\beta}^{\pi/\beta}
\biggl(\int_{-\pi}^{\pi}e^{\Lambda\Phi(\varphi,\psi)}
\,d\psi\biggr)\,d\varphi,
\end{equation}
where
\begin{gather}
\Lambda=\beta^{-3/2}e^{-\mu}\asymp N, \label{eq35:v851}
\\
\label{eq36:v851} \Phi(\varphi,\psi) =i\beta^{5/2}e^\mu M\varphi
+i\beta^{3/2}e^\mu N\psi +\beta^{3/2}\sum_{j=0}^{\infty}
q_je^{-\beta j-i(\psi+\beta j\varphi)}.
\end{gather}
(The sign $\asymp N$ means that there exist constants~$c_1$
and~$c_2$ such that $c_1N\le\lambda\le c_2N$).

Indeed, the substitution described above yields
\begin{align}
\cN_0(M,N)
&=\frac{\beta N!\mspace{2mu} e^{\beta M+\mu N}}{4\pi^2}
\int_{-\pi/\beta}^{\pi/\beta} \biggl(\int_{-\pi}^{\pi}
\sum_{\{N_{jk}\}} \frac{e^{i\beta\varphi M+i\psi N}}
{\prod_{j=0}^\infty\prod_{k=1}^{q_j}N_{jk}!}
\nonumber\\&\hphantom{{}\qquad\times \int_{-\pi/\beta}^{\pi/\beta}
\biggl(\int_{-\pi}^{\pi}}\qquad\qquad
\times\exp\biggl\{-\sum_{j=0}^\infty \sum_{k=1}^{q_j}N_{jk}
(\beta j+\mu+i\beta j\varphi+i\psi)\biggr\}\,
d\psi\biggr)\,d\varphi \nonumber\\& =\frac{\beta N!\mspace{2mu}
e^{\beta M+\mu N}}{4\pi^2} \int_{-\pi/\beta}^{\pi/\beta}
\biggl(\int_{-\pi}^{\pi}e^{i\beta\varphi M+i\psi N}
\prod_{j=0}^\infty\prod_{k=1}^{q_j} \sum_{N_{jk}=0}^\infty
\frac{e^{-N_{jk}(\beta j+\mu+i\beta j\varphi+i\psi)}}
{N_{jk}!}\,d\psi\biggr)\,d\varphi \nonumber\\& =\frac{\beta
N!\mspace{2mu} e^{\beta M+\mu N}}{4\pi^2}
\int_{-\pi/\beta}^{\pi/\beta} \biggl(\int_{-\pi}^{\pi}
e^{i\beta\varphi M+i\psi N} \prod_{j=0}^\infty\prod_{k=1}^{q_j}
\exp\{e^{-(\beta j+\mu+i\beta j\varphi+i\psi)}\}
\,d\psi\biggr)\,d\varphi \nonumber\\& =\frac{\beta N!\mspace{2mu}
e^{\beta M+\mu N}}{4\pi^2} \int_{-\pi/\beta}^{\pi/\beta}
\biggl(\int_{-\pi}^{\pi} \exp\biggl\{i\beta\varphi M+i\psi N
\sum_{j=0}^{\infty} q_je^{-(\beta j+\mu+i\beta
j\varphi+i\psi)}\biggr\} \,d\psi\biggr)\,d\varphi.
\label{eq37:v851}
\end{align}

\begin{lemma}
\label{l1:v851} The phase function $\Phi(\varphi,\psi)$ defined
by~\eqref{eq36:v851} possesses the following properties:
\begin{itemize}
\item[{\rm1.}]
All of its derivatives are uniformly bounded for the values
of~$\beta$ and $\xi$ satisfying inequality~\eqref{eq21:v851}.
\item[{\rm2.}]
The phase function has a stationary point
$\varphi=0\,\operatorname{mod}2\pi/b$,
$\psi=0\,\operatorname{mod}2\pi$.
\item[{\rm3.}] The matrix~$\Phi''(0,0)$
of second derivatives of the phase function at the stationary
point is nondegenerate and is strictly negative definite uniformly
in the parameters~$\beta$ and $\xi$ satisfying
inequality~\eqref{eq21:v851}.
\item[{\rm4.}]
the imaginary part of the phase function at the stationary point
is zero and its real part attains an absolute maximum there;
moreover, for any~$\gamma>0$, there exists a~$\delta>0$
independent of the parameters~$\beta$ and $\xi$ satisfying
inequality~\eqref{eq21:v851} such that
\begin{equation}
\label{eq38:v851} \operatorname{Re}{\Phi}(\varphi,\psi)
<\operatorname{Re}{\Phi}(0,0)-\delta\qquad \text{for}\quad
\operatorname{dist}((\varphi,\psi),(0,0))>\gamma.
\end{equation}
\end{itemize}
\end{lemma}

\begin{proof}
1. The boundedness of the derivatives of the phase function is
proved by direct calculations.

\smallskip

2. To to verify that the point $(0,0)$ is a stationary point of
the phase function, let us calculate its first derivatives:
\begin{align}
\label{eq39:v851} \frac{\partial\Phi}{\partial\varphi}
&=i\beta^{5/2}e^\mu\biggl[
 M-e^{-\mu}\sum_{j=0}^{\infty}jq_j
e^{-\beta j-i(\psi+\beta j\varphi)}\biggr],
\\
\frac{\partial\Phi}{\partial\psi} &=i\beta^{3/2}e^\mu\biggl[
 N-e^{-\mu}\sum_{j=0}^{\infty}q_j
e^{-\beta j-i(\psi+\beta j\varphi)}\biggr]. \label{eq40:v851}
\end{align}
 For
$\varphi=\psi=0$, both derivatives vanish by the definition of
the parameters~$\beta$ and $\xi=e^\mu$.

\smallskip

3. The matrix~$\Phi''(0,0)$ is of the form
\begin{equation}
\label{eq41:v851} \Phi''(0,0)=-\sum_{j=0}^{\infty}
\beta^{3/2}q_je^{-\beta j}
\begin{pmatrix}
1 &\beta j
\\
\beta j &\beta^2j^2
\end{pmatrix}.
\end{equation}
 Let us estimate
this matrix as the matrix of the corresponding quadratic form as
follows:
\begin{equation}
\label{eq42:v851} \Phi''(0,0)\le
-\sum_{j=[x_1/\beta]}^{[x_2/\beta]} \beta^{3/2}q_je^{-\beta j}
\begin{pmatrix}
1 &\beta j
\\
\beta j &\beta^2j^2
\end{pmatrix},
\end{equation}
where $x_2>x_1>0$ are arbitrary fixed numbers.

For small~$\beta$, in view of the asymptotics $q_j\simeq j/2$ for
large~$j$, the matrix on the right-hand side can be calculated by
the Euler--Maclaurin formula, obtaining as a result, up to~$o(1)$,
the matrix
\begin{equation}
\label{eq43:v851} -\frac12
\begin{pmatrix}
\displaystyle \int_{x_1}^{x_2}xe^{-x}\,dx &\displaystyle
\int_{x_1}^{x_2}x^{3/2}e^{-x}\,dx
\\[4mm]
\displaystyle \int_{x_1}^{x_2}x^{3/2}e^{-x}\,dx &\displaystyle
\int_{x_1}^{x_2}x^{2}e^{-x}\,dx
\end{pmatrix} =
-\frac12\begin{pmatrix} (1,1) &(1,x)
\\
(x,1) &(x,x)
\end{pmatrix},
\end{equation}
where
$$
(u,v)=\int_{x_1}^{x_2}u(x)\overline v(x)e^{-x}x\,dx
$$
is the inner product $L^2([x_1,x_2],e^{-x}x)$.

Since the functions~$1$ and $x$ are linearly independent, the
matrix~\eqref{eq43:v851} is negative definite, which proves the
required assertion.

\smallskip

4. For $\Phi(\varphi,\psi)$, from formula~\eqref{eq36:v851} we
obtain
$$
\operatorname{Re}{\Phi}(0,0)
-\operatorname{Re}{\Phi}(\varphi,\psi)
=\beta^{3/2}\sum_{j=0}^{\infty} q_je^{-\beta j}(1-\cos(\psi+\beta
j\varphi)).
$$
All the summands on the right-hand side are nonnegative.
Therefore, omitting part of them and estimating the coefficients
$\beta q_je^{-\beta j}$ for the remaining summands, we obtain
\begin{align}
\nonumber \operatorname{Re}{\Phi}(0,0)
-\operatorname{Re}{\Phi}(\varphi,\psi) &\ge\operatorname{const}
\beta\sum_{j=[x_1/\beta]}^{[x_2/\beta]} (1-cos(j\varphi+\psi))
\\&
\ge\operatorname{const} \biggl(x_2-x_1-\beta
-\beta\biggl|\sin\frac{\beta\varphi}2\biggr|^{-1}\biggr).
\label{eq44:v851}
\end{align}
Choosing~$x_1$ and $x_2$ in a suitable way, we obtain the
required assertion. The lemma is proved.
\end{proof}

Using this lemma, we can calculate the integral~\eqref{eq34:v851}
by the saddle-point method and obtain a lower bound for the
number of ordered samples in the form
\begin{equation}
\label{eq45:v851} \cN(N,M) \ge C\beta\Lambda^{-1}\sqrt N
\exp\{N\ln N+\beta M+\mu N\}.
\end{equation}

Now we can estimate the integral~\eqref{eq37:v851}. By
Lemma~\ref{l1:v851}, all the derivative of the
function~$\Phi(\varphi)$ are uniformly bounded. In addition, if
$\Phi$ is expressed in the form $\Phi=\Phi_1+i\Phi_2$,
where~$\Phi_1$ and $\Phi_2$ are real, then,
$$
\Phi_1'(0)=0,\quad \Phi_1''(0)<-C<0,\qquad
\Phi_2'(0)=\Phi_2''(0)=0.
$$
Hence, for $|\varphi|\le\varepsilon$, where $\varepsilon>0$ is
sufficiently small, using the Taylor formula with remainder, we
obtain the estimates
\begin{gather}
\label{eq46:v851} \Phi(0)-C_1|\varphi|^2 \le\Phi_1(\varphi)
\le\Phi(0)-C_2|\varphi|^2,
\\
\label{eq47:v851} |\Phi_2(\varphi)|\le C_3|\varphi|^3
\end{gather}
where the~$C_j$ are positive constants independent of~$M$ and the
sequence~$\{{1}R_j{1}\}$.

Suppose that
$$
1=\psi_1(\varphi)+\psi_2(\varphi)
$$
is a nonnegative smooth partition of unity on the circle
$S^1\ni\varphi$ of radius~$b$ such that
$$
\operatorname{supp}\psi_1 \subset[-\varepsilon,\varepsilon]\qquad
\text{and}\qquad\psi_1(\varphi)=1\quad \text{for}\ \
\varphi\in\biggl[-\frac\varepsilon2\,, \frac\varepsilon2\biggr].
$$
Let us express the integral
$$
I=\int_{S^1}\exp\{\beta^{-3/2}\Phi(\varphi)\}\,d\varphi
$$
as the sum
$$
I=\int_{S^1}\exp\{\beta^{-3/2}\Phi(\varphi)\}
\psi_1(\varphi)\,d\varphi
+\int_{S^1}\exp\{\beta^{-3/2}\Phi(\varphi)\}
\psi_2(\varphi)\,d\varphi\equiv I_1+I_2.
$$
By Lemma~\ref{l1:v851}, item~3,
$$
\operatorname{Re}\Phi(\varphi) \le\Phi(0)-\delta, \qquad \delta>0,
$$
on the support of the integrand in~$I_2$, while the measure of
the support is of the order of~$\beta^{-1}$.
 Therefore,
\begin{equation}
\label{eq48:v851} |I_2|\le K\exp\biggl\{\frac{\beta^{-3/2}\delta}2
+\beta^{-3/2}\Phi(0)\biggr\},\qquad b\to0,
\end{equation}
where~$K$ is a constant.

Let us now estimate the integral~$I_1$. For convenience, denote
provisionally by~$h=\beta^{3/2}$ the small parameter in the
exponential of our integral. On the interval
$D=[-\varepsilon,\varepsilon]$, we distinguish two subintervals
$D_{1/2}\subset D_{1/3}\subset D$ by setting
\begin{equation}
\label{eq49:v851}
 D_{1/2}=[-\varepsilon h^{1/2},\varepsilon h^{1/2}],\qquad
 D_{1/3}=[-\varepsilon h^{1/3},\varepsilon h^{1/3}].
\end{equation}
Then
$$
\biggl|\frac{\Phi_2}h\biggr|\le C_3\varepsilon^3
\qquad\text{for}\quad \varphi\in D_{1/3},
$$
so that (for a sufficiently small~$\varepsilon$) the imaginary
part of the argument of the exponential~$D_{1/3}$ is small and
the following relation holds:
\begin{equation}
\label{eq50:v851} \operatorname{Re}e^{\Phi(\varphi)/h}
\ge\frac12\mspace{1mu} e^{\Phi_1(\varphi)/h},\qquad \varphi\in
D_{1/3}.
\end{equation}
 Further,
\begin{equation}
\label{eq51:v851} \frac{\Phi(0)}h \ge\frac{\Phi_1(\varphi)}h
\ge\frac{\Phi(0)}h-C_1\varepsilon^2, \qquad \varphi\in D_{1/2}.
\end{equation}
Combining this with the previous inequality and taking into
account the fact that the length of the interval~$D_{1/2}$ is
equal to~$2\varepsilon h^{1/2}$, we obtain
$$
\operatorname{Re}\int_{D_{1/2}}
e^{\Phi(\varphi)/h}\psi_1(\varphi)\,d\varphi \ge
C_4e^{\Phi(0)/h}h^{1/2}.
$$
Further,
$$
\operatorname{Re}\int_{D_{1/3}\setminus D_{1/2}}
e^{\Phi(\varphi)/h}\psi_1(\varphi)\,d\varphi\ge0
$$
by virtue of~\eqref{eq50:v851}. Moreover, the following
inequality holds:
$$
\frac{\Phi_1(\varphi)}h
\le\frac{\Phi_1(0)}h-C_2\varepsilon^2h^{-1/3},\qquad \varphi\in
D\setminus D_{1/3},
$$
so that
$$
\biggl|\int_{D\setminus D_{1/3}}
e^{\Phi(\varphi)/h}\psi_1(\varphi)\,d\varphi\biggr| \le
C_5e^{\Phi(0)/h}e^{-C_2\varepsilon^2h^{-1/3}}.
$$
Combining all the previous estimates, we obtain
$$
I=\operatorname{Re}I \ge
C_6\beta^{3/4}\exp\{\beta^{-3/2}\Phi(0)\}.
$$

It remains to substitute this estimate into
formula~\eqref{eq37:v851} for $\cN_0(M,N)$ and, in view of the
formulas~\eqref{eq36:v851} for the phase function and by the
inequality $\cN(M,N)>\cN_0(M,N)$, we obtain a lower bound for
$\cN(M,N)$. As a result, we obtain Theorem~\ref{t1:v851}. Since
the number $\cN(M,N)$ corresponds to the Lebesgue measure of the
total phase volume and the number $\cN(M,N,\Delta)$ corresponds
to the Lebesgue measure of the phase volume defined in
parentheses in formula~\eqref{eq17:v851}, we obtain the proof of
Theorem~\ref{t2:v851}.

\begin{remark}
The Maxwell distribution  \eqref{eq1:v851}, \eqref{eq2:v851} holds
for a ''classical  ideal gas'' in common understanding. By
definition of pressure $P$  of specific volume $V_{\text{sp}} =
V/N$ for a ''classical  ideal gas'' the compressibility factor
$$
Z= \frac {PV_{\text{sp}}}{kT}
$$
is identically equal  to $1$.
\end{remark}

\section{Clusterization in an ideal gas and dependence of the compressibility
factor on the pressure}

Each scientist who refutes a century old theory runs the risk of
being accused of incompetence and of irritating those scientists
who absorbed the old theory ''with their mother's milk.''  And if
this is a scientist who has achieved a good deal in his area of
knowledge, he also runs the risk of losing his hard-earned
authority. This is borne out by the history of new discoveries in
physics. Thus, the great physicist Boltzmann, virulently attacked
by his contemporaries, committed suicide by throwing himself down
the well of a staircase.

In 1900, Planck proposed his famous formula describing black body
radiation, which gave results coinciding with experiments, but
which he had not rigorously established. The mathematician Bose
from India noticed that, in order to derive the formula, one must
use a new statistic instead of the old one, the so-called
Boltzmann or Gibbs statistic. It is possible that Planck was also
aware of this statistic, but was afraid of being criticized or
did not really believe in his own result. Bose, just like
Boltzmann, was the object of virulent criticism, until Einstein
gave his approval to the proposed statistic, which was also
justified by the philosophical concepts of Ernst Mach. At first,
physicists were bewildered and could not understand the Bose
statistic, because they could not imagine how moving particles
can exchange positions without using up any energy.

These two statistics have been illustrated above by a simple
financial example.

The reply to the bewilderment of physicists was given by Mach's
philosophical conception, claiming that the basic notions of
classical physics (space, time, motion) are subjective in origin,
and the external world is merely the sum of our feelings, and the
goal of science is to describe these feelings. Therefore, if we
are unable to distinguish particles in our subjective perception,
then they are undistinguishable.

I propose a completely different philosophy. We can regard
particles as distinguishable as well as undistinguishable. This
only depends on the aspect of the system of particles that we are
interested in, i.e., depends on the question we are seeking an
answer to. Thus, returning to the money example, people are
interested in the denominations of the bank notes they own, not
in their serial numbers (unless, of course, they believe in
''lucky numbers'').

The situation in physics is similar. Suppose we have a receptacle
filled with gas consisting of numerous moving particles. If we
take a slow snapshot of the gas, the moving particles will
display "tails" whose lengths depend on the velocity of the
particle: the faster the motion, the longer the tail. Using such
a photograph, we can determine the number of particles that move
within a given interval of velocities. And we don't care where
which individual particle is located and which particular
particle has the given velocity.

I have derived formulas which show how the number of particles is
distributed with respect to velocity, for example, they show for
what number (numerical interval) it is most probable to meet a
particle moving with a velocity in that interval.

These formulas lead to a surprising mathematical fact: there
exists a certain maximal number of particles after which the
formulas must be drastically modified. If the number of particles
is much less than this maximal number, the formulas coincide with
the Gibbs distribution up to multiplication by a constant.
Nevertheless, this is essential, because the corrected Gibbs
formula thus obtained no longer leads to the Gibbs paradox.

The paradox now bearing his name was stated by Gibbs in his paper
''On the equilibrium of heterogenous matter,'' published in
several installments in 1876-1879, and resulted in great interest
on the part of physicists, mathematicians, and philosophers. This
problem was studied by H. Poincare, G. Lorentz, J. Van-der-Waals,
V. Nernst, M. Planck, E. Fermi, A. Einstein, J. von Neumann, E.
Schrodinger, I. E. Tamm, P. V. Bridgeman, L. Brillouin, A. Lande
and others, among them nine Nobel Prize laureates.

From my point of view, the solution of the Gibbs paradox can be
obtained once we realize that the Gibbs formula in its classical
form is invalid and we modify it in the way that I have
indicated. This modification was previously interpreted as a
consequence of quantum theory, but this is erroneous from the
mathematical point of view, since the passage from quantum
mechanics to classical mechanics cannot change symmetry and
therefore cannot change the statistics.

In this situation, the following phenomenon, rather strange from
the mathematical point of view, arises. If the number of
particles is greater than the maximal number indicated above,
then the "superfluous" particles, as we already explained, do not
fit into the obtained distribution and the velocity of these
particles turns out to be much less than the mean velocity of
particles in the gas. This effect differs from the Bose-Einstein
condensate phenomenon from quantum theory, because in quantum
theory these particles are at the very lowest energy level, they
have the lowest speed, i.e., roughly speaking, they stop.

Further, I try to give a physical interpretation to the obtained
rigorous mathematical formulas. I interpret the maximal number of
particles mentioned above as oversaturated vapor; the superfluous
particles are then regarded as nuclei around which droplets begin
to grow. As a result, this can explain the so-called phase
transition of the first kind, in which, as the result of the
system achieving equilibrium, the number of particles changes
from that number for an oversaturated gas to that for a saturated
one. Indeed, it is only those particles which move at speeds
greater than the speed of the "superfluous" particles that can be
doubtlessly regarded as particles of the "pure" gas (vapor),
while the others have condensed or have mixed with the condensed
particles (clusters).

In section 1, I cited an example from economics, similar to the
one above, that supports exchangeability theory (instead of the
''independence condition''). In my opinion, we must revise, in
this vein, the ''Gibbs conjecture on thermodynamic equilibrium,''
which is based (see \cite{2:v851}) on the property of independence
leading to the theorem on the multiplication of probabilities. It
is this conjecture that leads to the Gibbs distribution, which is
refuted by the Gibbs paradox, i.e., in essence, by the
mathematical counterexample to this conjecture, as mentioned
above.

It is difficult for physicists to grasp this problem, because it
involves a mathematical effect of the type of Bose condensation,
which results in the appearance of a ''Bose condensate,'' which,
from the author's point of view, has been treated  as some
coagulation of particles with low velocities and the  formation
of dimers, trimers, and other clusters.

The phenomenon of the appearance of dimers is usually obtained by
modeling involving the initial conditions and interactions, for
example, of Lennard-Jones type. According to the author's point
of view, if this phenomenon involves interaction, then it can
occur before the switching-on of an interaction of Lennard-Jones
type: as far as the specific volume is concerned, we still deal
with an ideal gas.  Such type of interaction is observed, for
example, in the gas~$C_{60}$ (fullerene) possessing very weak
attraction (of order $O(1/r^9)$).  It is related to the asymmetry
of the molecules and the types of adjoining faces of the
molecules of~$C_{60}$.

This is much easier to observe experimentally, because fullerene
has no liquid phase and is immediately transformed into fullerite
particles.

The presence of such a ''saturated'' total number of particles in
the problem under consideration, with surplus particles going
somewhere (passing into the Bose condensate
\footnote{The physicists to whom I described this theory warned
me not to use the term ``Bose condensate,'' because this evokes
associations obscuring the understanding of the proposed
theory.}), is a mathematical fact rigorously proved together with
clear estimates of where such aggregates may occur.  However, it
is not quite correct to say that the particles are added.  Indeed,
it is better to say that we lowered the temperature, while using
a piston to maintain a constant pressure, and hence the saturated
total number of particles is decreased.  And we can simply say
that, for a given temperature, the pressure is increased until
the $\lambda$-transition occurs in the ``Bose condensate.'' The
question is: Where have the other particles gone if the
temperature is lowered simultaneously with the pressure
limitation or the pressure at the given temperature becomes
sufficiently large?  Perhaps, they precipitate on the walls of
the vessel?  Such a law of ``necessary'' precipitation
(coagulation) on the walls would be more interesting still and
would have important practical applications.  However,
experiments tend to support, to a greater extent, the first point
of view.  The physicists are even of the opinion that the
transition to dimers is a phase transition.

The most significant fact is that this estimate not improvable.
This fact follows from Theorem~2 in~\cite{11:v851}. In the case
of saturation, it makes it possible to determine the number of
particles passing into clusters as the temperature is lowered,
while a constant pressure is maintained by a piston
(see~\cite{12:v851}).  This also solves the Gibbs paradox.

The use of an unordered sample with replacement leads us to a
mathematical formula for the Bose gas, however, without the
parameter $\hbar$, the Planck constant, but with the same
parameters that appeared when using the parameters of the
Lennard-Jones interaction potential. Instead of
formula~\eqref{eq16:v851}, we thus obtain
\begin{equation}
\label{eq52:v851} \cN_{\Delta v}
=\frac{1}{2{.}612}\int_{v_1}^{v_1+\Delta v} \biggl(\frac{m}{2\pi
kT}\biggr)^{3/2} \frac{1}{e^{mv^2/2kT}-1}4\pi v^2\, dv.
\end{equation}
However, $\Delta v$ now depends on~$v_1$ in the following way: if
$v_1\sim1$, then $\Delta v=N_0^{-1/2+\delta}$, where $\delta>0$,
and $N_0$ is the number of particles saturating the volume~$V$ at
temperature~$T$ and maximal energy~$E/U_0$; namely,
\begin{equation}
\label{eq53:v851}
 N_0=\biggl(\frac{m}{U_0}\biggr)^{3/2}
\int_0^\infty\frac{4\pi u^2\,du } {e^{(m u^2)/2kT}-1}.
\end{equation}

In view of the given parameters, the velocity can be expressed as
$v=\sqrt{U_0/2m}$. Suppose that $v_0$ is the minimal velocity; it
is equal to $v_0=vN^{-1/3+\delta}$, where $\delta>0$ determines
the smallness of~$v_0$. For $\delta=1/3$, we obtain $v_0=v$;
therefore, we set $1/3>\delta>0$. The estimate of the error in
the formula for the distribution~\eqref{eq52:v851} is of the form
$$
O(N^{-1/3-\delta/2}\sqrt{\ln N}\mspace{2mu} (\ln\ln N)^\varepsilon);
$$
namely, the following theorem is valid.

\begin{theorem}
\label{t3:v851}
 The following relation holds:
\begin{align}
\nonumber &\mathsf{P}\Biggl(\cN_{\Delta v}
-\frac{1}{2{.}612}\int_{|v_0|}^{|v_0| +O(N_0^{-1/3+\delta_1})}
\biggl(\frac{m}{2\pi kT}\biggr)^{3/2}
\frac{1}{e^{mv^2/2kT}-1}\,dv_x\,dv_y\,dv_z
\\
\nonumber &\hphantom{\mathsf{P}\Biggl(}\qquad \ge
N^{-1/3-\delta/2}\sqrt{\ln N}\mspace{2mu} (\ln\ln
N)^\varepsilon\Biggr)
\\&\qquad
\le O(N_0^{-k}), \label{eq54:v851}
\end{align}
where $k$ is any integer, $\varepsilon>0$ is arbitrarily small,
$v_0\ge0$, $0<\delta<1/3$, $\delta_1>\delta$, $v=\sqrt{U_0/2m}$,
and $\Delta v=O(N_0^{-1/3+\delta_1}$.
 Here
$\mathsf P$ is is the Lebesgue measure of the phase volume defined
in parentheses in \eqref{eq54:v851} with respect to the total
volume.
\end{theorem}

The proof of  Theorem~\ref{t3:v851} is similar to the proof of
Theorem~\ref{t2:v851} except that we use  unordered samples of
''balls''  with replacement.

Suppose that there is a sequence of boxes $U_j$, $j=0,1,2,\dots$,
and each box $U_j$ is divided into $q_j$ compartments. We take $N$
identical balls and put them into the boxes at random observing
the only condition that
\begin{equation}\label{e1}
    \sum_{j=0}^\infty jN_j\le M,
\end{equation}
where $N_j$ is the number of balls in the box $U_j$ and $M$ is a
positive integer specified in advance. As an outcome, we obtain a
sequence of nonnegative integers $N_j$, $j=0,1,2,\dots$, such that
\begin{equation}\label{e2}
    \sum_{j=0}^\infty N_j=N
\end{equation}
and condition \eqref{e1} is satisfied. It is easily seen that,
given $M$ and $N$, there are finitely many such sequences. Suppose
that all allocations of balls to compartments are equiprobable.
Since the number of ways to distribute $N_j$ indistinguishable
balls over $q_j$ compartments is equal to
\begin{equation}\label{a}
    \binom{q_j+N_j-1}{N_j}
    =\frac{\Gamma(q_j+N_j)}{\Gamma(N_j+1)\Gamma(q_j)}
\end{equation}
(where $\Gamma(x)$ is the Euler gamma function), it follows that
each sequence $\{N_j\}$ can be realized in $f(\{N_j\})$ ways,
where
\begin{equation}\label{e3}
  f(\{N_j\})=\prod_{j=0}^\infty
        \frac{\Gamma(q_j+N_j)}{\Gamma(N_j+1)\Gamma(q_j)},
\end{equation}
and the probability of this sequence is equal to $f(\{N_j\})$
divided by the sum of the expressions similar to \eqref{e3} over
all sequences of nonnegative integers satisfying the constraints
\eqref{e1} and \eqref{e2}. This makes the set of all such
sequences a probability space; the corresponding probabilities
will be denoted by $\mathsf{P}(\cdot)$. The numbers $q_j$ are
called the \textit{multiplicities}. We shall assume that $q_0$ is
some positive integer and
\begin{equation}\label{e4}
    q_j=[j^{1/2}], \quad j=1,2,\dotsc,
\end{equation}
where the brackets stand for the integer part of a number.

What happens as $M,N\to\infty$? It turns out that the so-called
\textit{condensation phenomenon} occurs: if $N$ tends to infinity
too rapidly, namely, if $N$ exceeds some threshold
$N_{cr}=N_{cr}(M)$, then a majority of the excessive $N-N_{cr}$
balls end up landing in the box $U_0$; more precisely, with
probability asymptotically equal to $1$, the number of balls in
$U_0$ is close to $N-N_{cr}$ (and accordingly, the total number of
balls in all the other boxes is close to $N_{cr}$, now matter how
large $N$ itself is). Let us give  the scheme of proof analogous
to the proof of  Theorem~\ref{t2:v851}

Define $N_{cr}=N_{cr}(M)$ by the formula
\begin{equation}\label{e5}
    N_{cr}=\sum_{j=1}^\infty\frac{q_j}{e^{\beta j}-1},
\end{equation}
where $b$ is the unique positive root of the equation
\begin{equation}\label{e6}
    \sum_{j=1}^\infty\frac{jq_j}{e^{\beta j}-1}=M.
\end{equation}
Next, let
\begin{equation}\label{e7}
\Delta=N_{cr}^{2/3+\varepsilon},
\end{equation}
where $\varepsilon>0$ is arbitrarily small \textup(but
fixed\textup). If $N>N_{cr}$, then there exist constants $C_m$
such that
\begin{equation}\label{e8}
    \mathsf{P}(\vert N_0-(N-N_{cr})\vert>\Delta)\le C_mN_{cr}^{-m},
    \qquad m=1,2,\dotsc\,.
\end{equation}

It is not hard to compute $N_{cr}(M)$. Indeed, in view of
\eqref{e4}, the Euler--Maclaurin formula gives
\begin{equation}\label{e9}
    \sum_{j=1}^\infty\frac{jq_j}{e^{\beta j}-1}\sim
    \beta^{-5/2}\int_0^\infty
    \frac{x^{3/2}\,dx}{e^x-1}=\beta^{-5/2}\Gamma(\frac 52)\zeta(\frac 52)
\end{equation}
(where $\zeta(x)$ is the Euler zeta function) and likewise,
\begin{equation}\label{e10}
    \sum_{j=1}^\infty\frac{q_j}{e^{\beta j}-1}
    \sim \beta^{-3/2}\int_0^\infty
    \frac{x^{1/2}\,dx}{e^x-1}=\beta^{-3/2}\Gamma(\frac 12)\zeta(\frac 12).
\end{equation}
By substituting this into \eqref{e5} and \eqref{e6}, we obtain
\begin{equation}\label{e11}
    N_{cr}\sim \frac{M^{3/5}\Gamma(\frac 32)\zeta(\frac 32)}
    {(\Gamma(\frac 52)\zeta(\frac 52))^{3/5}}.
\end{equation}

In contrast to the Maxwell distribution,
the compressibility factor for the given distribution equals
\begin{equation}\label{z1}
Z=\frac{PV}{kTN_0}= \frac 23 \frac{\int\frac{(p^2/2m) p^2
dp}{e^{p^2/2mkT}-1}}{kT \int\frac{p^2 dp}{e^{p^2/2mkT}-1}}= 0.523.
\end{equation}
However, if  the number of particles $N \ll N_0$, then
\begin{equation}\label{z2}
Z=\frac{PV}{N}= \frac 23\frac{\int\frac{(p^2/2m) p^2
dp}{e^{(p^2/2m -\mu)/ kT}-1}}{kT \int\frac{p^2
dp}{e^{(p^2/2m-\mu)/kT}-1}}
\end{equation}
i.e., there appears a negative parameter $\mu$ which tends to $-\infty$
(and $Z\to1$) as $N/N_0$ decreases to zero.

Consider the gas which consists of $K$  different molecules, or
dimers, trimers, \dots,  $k$-mers.

Now suppose that the situation is the same, but we should
additionally paint each of the $N$ balls at random into one of $K$
distinct colors. Now that we can distinguish between balls of
different colors but balls of a same color are indistinguishable,
how does this affect the probabilities?

Instead of immediately painting the balls, we can further divide
each of the $q_j$ compartments in the $j$th box into $K$
sub-compartments and put the uncolored balls there (with the
understanding that the balls in the $k$th sub-compartment will
then be painted into the $k$th color and the dividing walls
between the sub-compartments will be removed). Now we have $Kq_j$
sub-compartments in the $j$th box, so that there are
\begin{equation}\label{b}
    \binom{Kq_j+N_j-1}{N_j}
    =\frac{\Gamma(Kq_j+N_j)}{\Gamma(N_j+1)\Gamma(Kq_j)}
\end{equation}
ways to put $N_j$ balls into the $j$th box. All in all, the
introduction of $K$ colors has the only effect that all
multiplicities $q_j$ are multiplied by $K$.

Our theorem applies in the new situation (with $q_j$ replaced by
the new multiplicities $\widetilde q_j=Kq_j$). The computation of
the new threshold $\widetilde N_{cr}$ mimics that of $N_{cr}$,
with the factor $K$ taken into account:
\begin{align*}
    \sum_{j=1}^\infty\frac{jKq_j}{e^{\beta j}-1}&\sim
    K\beta^{-5/2}\int_0^\infty
    \frac{x^{3/2}\,dx}{e^x-1}=K\beta^{-3/2}\Gamma(\frac 52)\zeta(\frac 52),\\
    \sum_{j=1}^\infty\frac{Kq_j}{e^{\beta j}-1}
    &\sim K\beta^{-3/2}\int_0^\infty
    \frac{x^{1/2}\,dx}{e^x-1}=K\beta^{-3/2}\Gamma(\frac 32)\zeta(\frac 32),\\
    \widetilde N_{cr}&\sim K\frac{M^{3/5}\Gamma(\frac 32)\zeta(\frac 32)}
    {(K\Gamma(\frac 52)\zeta(\frac 52))^{3/5}}=
    K^{2/5}N_{cr},
\end{align*}
where $N\ge N_{cr}+\Delta$.

Consider the following auxiliary problem: we wish to put some
balls into the boxes $U_j$, $j=1,2,\dotsc$, of multiplicities
$q_j$, leaving the box $U_0$ aside. The overall number of balls
is not specified in advance, and we should only observe the
condition
\begin{equation}\label{w1}
    \sum_{j=1}^\infty jN_j\le M.
\end{equation}
Theorem 10 in \cite{MaslNaz_83-2} and Theorem 1 in
\cite{MaslNaz_83-3} claim that in this problem the sum of all
$N_j$ is in most cases close to $N_{cr}$. More precisely, one has
the estimate
\begin{equation}\label{w2}
    \mathsf{P}\biggl(\biggl|N_{cr}
    -\sum_{j=1}^\infty N_j\biggr|>\Delta\biggr)
    \le C_mN_{cr}^{-m}
\end{equation}
with some constants $C_m$, $m=1,2,\dotsc$.

Let $G(L)$ be the number of ways to put exactly $L$ balls into the
boxes $U_j$, $j=1,2,\dots$, so that condition~\eqref{w1} is
satisfied. Note that $G(L)=0$ for $L>M$, because
\begin{equation}\label{w3}
    \sum_{j=1}^\infty jN_j\ge \sum_{j=1}^\infty N_j=L.
\end{equation}
Then the estimate \eqref{w2} can be rewritten as
\begin{equation}\label{w4}
    \frac{\sum_{|\alpha-N_{cr}|>\Delta} G(L)}{\sum_L G(L)}
    \le C_mN_{cr}^{-m}.
\end{equation}
Let $\mathcal{N}$ be the total number of ways to put $N$ balls
into the boxes $U_0,U_1,\dotsc$ with condition \eqref{e1} being
satisfied, and let $\mathcal{N}(\Delta)$ be the number of only
those ways for which, in addition,
\begin{equation}\label{w5}
    \vert N_0-(N-N_{cr})\vert>\Delta.
\end{equation}
One obviously has
\begin{equation}\label{w6}
    \mathcal{N}=\sum_{L=0}^{N'}G(L)F(N-L),
\end{equation}
where $F(x)$ is the number of ways to put $x$ balls into the box
$U_0$ of multiplicity $q_0$ and $N'=\min\{N,M\}$. In a similar
way,
\begin{equation}\label{w7}
    \mathcal{N}(\Delta)
    =\sum_{\substack{0\le L\le N'\cr
    |L-N_{cr}|>\Delta}}G(L)F(N-L).
\end{equation}
Note that $F(x)$ is a monotone increasing function. Hence we can
estimate
\begin{equation}\label{w8}
    \mathcal{N}\ge F(N-N')\sum_{L=0}^{N'}G(L)
    \ge\frac12 F(N-N')\sum_{L=0}^{M}G(L).
\end{equation}
(The last inequality follows from \eqref{w4} and (ii).) Next,
\begin{equation}\label{w9}
    \mathcal{N}(\Delta)\le
    F(N)\sum_{\substack{0\le L\le N'\cr
    |L-N_{cr}|>\Delta}}G(L)
    \le F(N)\sum_{|L-N_{cr}|>\Delta}G(L).
\end{equation}
By dividing \eqref{w9} by \eqref{w8}, we obtain
\begin{equation}\label{w10}
    \frac{\mathcal{N}(\Delta)}{\mathcal{N}}\le
    2\frac{F(N)}{F(N-N')}\frac{\sum_{|L-N_{cr}|>\Delta} G(L)}{\sum_L G(L)}
    \le 2C_mN_{cr}^{-m}\frac{F(N)}{F(N-N')}
\end{equation}
in view of \eqref{w4}. It remains to note that $F(x)\sim
Cx^{q_0-1}$ with some constant $C>0$, and hence
\begin{equation}\label{w11}
   \frac{F(N)}{F(N-N')}\le C_0\biggl(\frac{N}{N-N'}\biggr)^{1/2}
   \le C_1 M^{1/2}\le  {C_2}N_{cr}^{5/6}.
\end{equation}
By substituting this into \eqref{w10}, we obtain the desired
estimate. The proof of the proposition is complete.

At $K > 1$ the chemical potential $\mu$ in \eqref{z2} is strictly
less than zero, hence the compressibility factor will be greater
than the value of \eqref{z1}.

First, consider the graphs in Figs.~\ref{fig1} and ~\ref{fig2}
for argon.

If the vapor is saturated, then, at low temperatures, the number
of clusters (dimers, trimers) is, as a rule, large. This decreases
the total number of particles in the volume and increases the
chemical potential, and hence the compressibility factor $Z=
PV_{\text{sp}}/kT$, $P$ - is the pressure, $V_{\text{sp}}$ is the
specific volume, is increased. As the temperature increases, the
number of clusters decreases and, at a certain temperature, the
fraction of dimers becomes less than 7\% (the Calo criterion).
Then the compressibility can drop to 0.53.

But since the saturated gas is in equilibrium with the liquid,
the dimension can then decrease rather steeply and the
compressibility  factor (e.g., for argon) can decrease down to
0.25. It means that as the pressure increases, interaction takes
effect.

\begin{figure}
\includegraphics{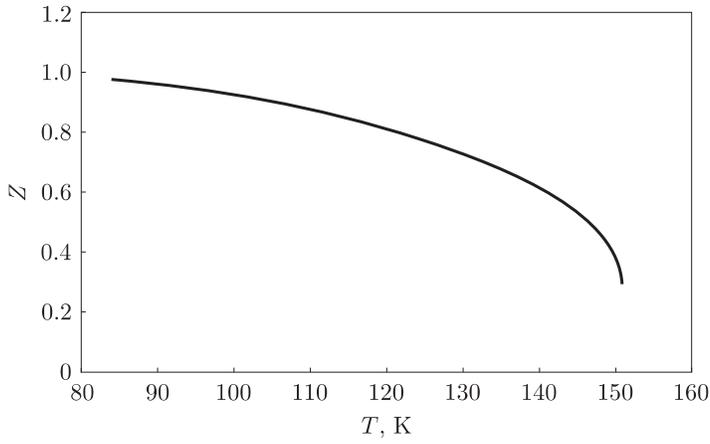}
\caption{Thermodynamic properties of saturated argon. $Z$ is the
compressibility factor, $Z=PV/kT$; \  $T$   is the temperature in
Kelvin degrees.} \label{fig1}
\end{figure}

\begin{figure}
\includegraphics{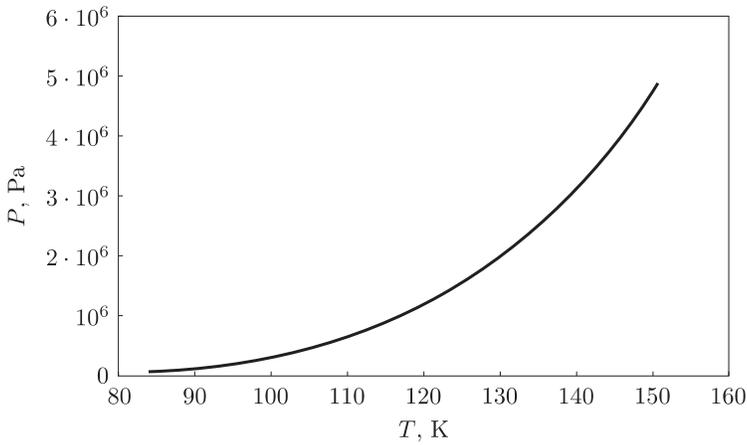}
\caption{Thermodynamic properties of saturated argon. P is the
pressure in pascals, T  is the temperature in Kelvin degrees.}
\label{fig2}
\end{figure}

Thus, the formation of nanostructures in the other phase (the
liquid one) plays a significant role, just as the formation of
clusters in a gas.

Let us now pass to the case of a constant temperature
(Fig.~\ref{fig3})\cite{Enciclop}.
\begin{figure}
\includegraphics{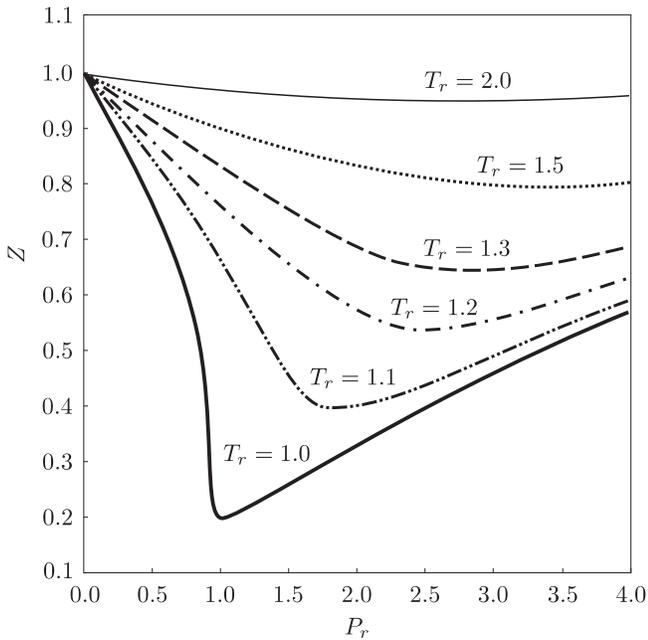}
\caption{ $T_r = T/T_c$, and $P_r = P/P_c$ are reduced temperature
and pressure, respectively.} \label{fig3}
\end{figure}

We can assume that, instead of $E$, $V^{2/3}$ tends to infinity.
And hence, in all the formulas of Bose-Einstein type from [2],
[3], we can assume that $b^{-1}\simeq \frac{V^{2/3}}{kT}$, where
$T$ is the temperature, $k$ is the Boltzmann constant, and $d$ is
the dimension. For the number $b$ in \cite{Masl-MZ-83-5,
Masl-MZ-83-6} to be dimensionless, let us introduce the effective
radius $a$ of the gas molecule (see below).

\section{Taking into account the pair interactions between particles}

Now consider the Hougen--Watson diagram given in Brushtein's
textbook ``Molecular Physics''~\cite{30:v851}. The diagram
reflects the dependence of the compressibility factor $Z=PV/kTN$
on the pressure for different temperatures and was constructed by
Hougen and Watson for seven gases: H$_2$, N$_2$, CO, NH$_3$,
CH$_4$, C$_3$H$_8$, C$_5$H$_{12}$. Although the textbook states
that attraction decreases compressibility, this, however, is
obtained for Van-der-Waals gas under the condition $|1-Z|\ll1$,
i.e., time as a compressibility factor decreases down to~$0{.}2$.

From our understanding of ideal gas, it follows from the
distribution~\eqref{eq52:v851} that $Z$ can attain the value
of~$0{.}523$ and, further, the compressibility factor must
decrease only at the expense of the interaction.

Phenomenological thermodynamics is based on the concept of pair
interaction.  Moreover, it is implicitly assumed that

\textit{there exists some one-particle distribution
characterizing the field, to which all the particles contribute}.

\begin{figure}
\includegraphics{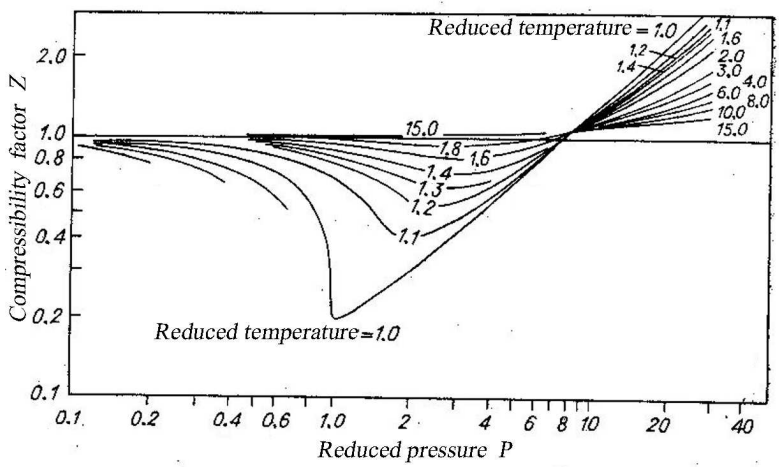}
\caption{The original Hougen-Watson diagram.} \label{fig4}
\end{figure}

They are interrelated. Formulas for the distribution corresponding
to this mean field were rigorously obtained by the author
in~\cite{13:v851}. The equation that relates the potential of the
mean field to pair interactions is called the equation of
self-consistent (or mean) field.  For the interaction
potential~$\Phi$, it is of the form
\begin{equation}
\label{eq55:v851} u(x,p)=u_0(x,p)+N\int\Phi(x-x',p-p')
\frac{1}{e^{(p^2/2m+u(x',p'))(kT)^{-1}}-1}\,dx'\,dp'.
\end{equation}
This equation was rigorously justified only in the case of
long-range interaction, in particular, in~\cite{14:v851}.

In the case of a gas occupying the volume~$V$ and not subject to
the action of external forces, $u_0(x,p)=0$ for~$x$ lying inside
the volume~$V$ and $u_0(x,p)=\infty$ on the boundary of this
volume.

In what follows, we shall study only this case. For the
interaction potential we take the Lennard-Jones potential or, for
a greater coincidence with the experiment, the following
potential:
\begin{equation}
\label{eq56:v851} \Phi(r)=\frac{\varepsilon n}{n-6} \biggl(\frac
n6\biggr)^{6/(n-6)} \biggl(\frac{\sigma^n}{r^n}-
\frac{\sigma^6}{r^6}\biggr)
\end{equation}
containing one more parameter $n>6$.

In the zeroth approximation, as $\sigma\to 0$, the integral of
the potential \eqref{eq55:v851} with respect to $r$ from some
$r_0$ to $\infty$. This integral substantially depends on $r_0$.
How must we choose $r_0$?

Consider the scattering of one particle by another. Suppose that,
as $t \to -\infty$, the velocities of the two colliding particles
are equal to ${\bf{v}}_1^{in}$ and ${\bf{v}}_2^{in}$,
respectively. This means that, at $t \to -\infty$  the
trajectories of the particles approach straight lines. In terms
of the variable $r=r_2-r_1$ as $t \to -\infty$  the radius vector
of the $r$-point asymptotically approaches the function
$r^{in}=\rho+{\bf{v}}^{in}t$, where $\rho{\bf{v}}^{in}=0$ and
${\bf{v}}^{in}= {\bf{v}}^{in}_2-{\bf{v}}^{in}_1$. The constant
vector $\rho$ is referred to as the target parameter. The
quantity $\rho$ is equal to the distance between the straight
lines along which the particles would move if no interaction was
present. After the collision as $t \to \infty$, the velocities of
the particles are equal to ${\bf{v}}_1^{out}$ and
${\bf{v}}_2^{out}$. This means that the radius vector $r(t)$
asymptotically approaches the function $r^{out} = c +
{\bf{v}}^{out}t$. The trajectories $r^{in}(t)$ and $r^{out}(t)$,
which are straight lines, are said to be the incoming and
outgoing asymptotes. The value of the relative speed in the
\textit{in}- and \textit{out}-states is preserved; namely,
$|{\bf{v}}^{in}|=|{\bf{v}}^{out}|=v.$

The condition on the turning point $r_0$ is of the form
\begin{equation}\label{r0}
\left[\frac{v^2}{4}- \frac{(\rho v)^2}{2r^2}-\Phi(r)\right]=0.
\end{equation}
Let is find the value of $\Phi(r)$ for the potential
\eqref{eq55:v851} (this value depends on $r_0$),
\begin{eqnarray}
&& \widetilde{\Phi}(0) = \int^\infty_{r_0} \Phi(|r|) dr=
4\varepsilon\int_{r_0}^\infty \left(
\frac{\sigma^{12}}{r^{12}}-\frac{\sigma^6}{r^6}\right) dxdydz =
\frac {16}{3} \pi\varepsilon \sigma^3\int^\infty_{r/\sigma}
\left(\frac{1}{\xi^{10}} - \frac{1}{\xi^4}\right) d\xi = \nn \\
&& = \frac{16\pi}{3}\varepsilon \sigma^3\left\{-\frac
19\left(\frac{\sigma}{r_0}\right)^9+\frac
13\left(\frac{\sigma}{r_0}\right)^3\right\} =
\frac{16}{3}\pi\varepsilon \sigma^3\left(-\frac
13\frac{\sigma^3}{r^3_0}+ \frac 19\frac{\sigma^9}{r_0^9}\right),
\label{phi}
\end{eqnarray}
where $r/a=\xi$.

In particular, for $r_0 = \sigma$, the area is equal to
$-(8/9)\varepsilon \sigma^3\alpha$, where $\alpha = (4/3)p$. For
$r_0=\frac{\sigma}{\sqrt[6]{3}}$, the area is equal to zero.

In the nanotube, the target parameter $\rho$ can be assumed to be
zero. In this case,
$$
\xi^{12} - \xi^6 = \frac{p^2}{4\varepsilon m}; \quad \xi=\frac
{\sigma}{r_0}
$$
$$
\xi^6=x
$$
$$
x^2-x-\frac{p^2}{4\varepsilon m}=0
$$
$$
x_{1,2}= \frac 12+\sqrt{\frac 14+\frac{p^2}{4\varepsilon m}} =
\frac 12 + \frac 12\sqrt{1+\frac{p^2}{\varepsilon m}}.
$$
\begin{equation}\label{a1}
\frac{\sigma}{r_0}= \left[\frac
12\left(1+\sqrt{1+\frac{p^2}{\varepsilon m}}\right)\right]^{1/6}
\end{equation}
for $p=0, r_0=\sigma$.

The Lennard-Jones potential can now be represented in the form
$\Phi(r_0, r)=\Phi(r_0(p),r)$, and hence, since $|p|=|p_i-p_j|$,
it follows that the equation for the dressed potential looks as
follows:
\begin{equation}\label{a3}
u(p, x) =N\int \Phi(r_0(p-\eta), |x-\xi|)\frac{dp
d\xi}{e^{\frac{b\eta^2}{2m}+u(\eta,\xi)}-1}, \quad p\in R^3, x\in
R^3.
\end{equation}
Since the external potential is absent and the distribution
depends on $x$ in terms of the dressed potential only, we can
assume that $u(p, x) = u(p)$ does not depend on $x$.

Making the change $x-\xi=y$ and integrating with respect to $y$
from $r_0$ to $\infty$, we obtain
\begin{equation}\label{a4}
u(p)= \frac{16}{3}\pi\varepsilon \sigma^3 \int  \frac{\frac
19\left[\frac 12+ \frac 12\sqrt{1+\frac{(p-\eta)^2}{\varepsilon
m}}\right]^{3/2}-\frac
13\left[\frac12+\frac12\sqrt{1+\frac{(p-\eta)^2}{\varepsilon
m}}\right]^{1/2}}{e^{b\{\frac{\eta^2}{2m}- \mu+u(\eta)\}}-1}
d\eta.
\end{equation}
Here $\mu\leq 0$ stands for the chemical potential.

It should be noted that the probability of the event in which the
particle $x_2$ occurs on the interval from $r_0$ to infinity is
not constant. It is obviously proportional to the time during
which the particle is kept within the interval ($x_2', x_2''$),
and this time is inversely proportional to the speed of $x_2$ with
respect to the particle $x_1$. One can readily see that this
probability is equal to
\begin{equation}\label{eq11}
F(x_2)= \frac{\left[\frac 1m (p_1-p_2)^2 -
\frac{\rho^2(p_2-p_1)^2}{m(x_2-x_1)^2}
-\Phi(x_2-x_1)\right]^{-1/2} - (\frac 1m(p_1 - p_2)^2)^{-1/2}}
{\int_{r_0}^\infty \left\{\left[\frac 1m (p_1-p_2)^2 -
\frac{\rho^2(p_2-p_1)^2}{m(x_2-x_1)^2}
-\Phi(x_2-x_1)\right]^{-1/2} - (\frac 1m(p_1 - p_2)^2)^{-1/2}
\right\}dr}.
\end{equation}
Since the integral of the expression
$$
\frac{F(x_2)}{e^{b\{\frac{p_2^2}{m}+u(p_2,
\rho)-\mu\}}-1}\left[\int \left(e^{b\{\frac{p^2_2}{m}+u(p,
\rho)-\mu\}}-1\right)^{-1} |p_2|^{2} dp_2\right]^{-1},
$$
(first with respect to $x_2$ and then with respect to $dp$) is
equal to one, it follows that the probabilities are independent,
and the distribution with respect to $x_2$ and $p_2$ is equal to
the product of the distributions with respect to $x_2$ and to
$p_2$.

Let us find the energy level below which the ``condensate''
appears.

As is well known, the ``turning point''~$r_0$, the energy
$E=m(v_1-v_2)^2$, where $v_1$ and $v_2$ are the velocities of two
interacting particles, and the impact parameter~$\rho$ are
related by
\begin{equation}
\label{eq63:v851}
 E-\frac{E\rho^2}{r_0^2}-\Phi(r_0)=0.
\end{equation}

\smallskip

1. The potential has the form $-\alpha/r^4$. Then the turning
point~$r_0$ in the scattering problem is defined by the relation
$$
1-\frac{\rho^2}{r^2_0}+\frac\alpha{r^4E}=0,
$$
where $E=(p_1-p_2)^2/m$ is the energy of the particles an
infinite distance apart and $\rho$ is the impact parameter. Hence
$$
r_0=\frac{\rho^2}2 +\sqrt{\frac{\rho^4}4-\frac\alpha
E}\mspace{2mu},
$$
and the solution is only possible if~$E$ is bounded below:
$E\ge4\alpha/\rho^4$.

For $E=4\alpha/b^4$, we have the expression
$$
1-\frac{b^2}{r^2} +\frac\alpha{r^4E}
=\biggl(1-\frac{\rho^2}{2r}\biggr)^2>0
$$
where $r>\rho^2/2$.

\smallskip

2. Suppose that the attraction potential is of the form
$$
\Phi(r)=-4U_0\mspace{1mu}\frac{\sigma^6}{r^6}\,, \qquad
 E_{\min}
=\min E=4U_0\sigma^6\max_{r_0<\rho}
\frac1{r_0^6}\biggl(\frac{\rho^2}{r_0^2}-1\biggr)^{-1}
=27U_0\mspace{1mu}\frac{\sigma^6}{\rho^6}\,.
$$

Then we obtain
$$
\frac{\Phi(r)}{E_{\min}} =-\frac4{27}\cdot\frac{\rho^6}{r^6}\,.
$$
Note that, for $E<E_{\min}$, if there is no term corresponding to
repulsion, then the ``falling on the center'' phenomenon occurs,
i.e., the binding-together of the particles.

\smallskip

3. For the function
$$
\Phi=4U_0\biggl(\frac{\sigma^6}{r^6}
-\frac{\sigma^{12}}{r^{12}}\biggr)
\biggl(\frac{\rho^2}{r^2}-1\biggr)^{-1},
$$
the equation for extremum points is of the form
$$
\Phi'(x)=8\varepsilon \sigma^6\mspace{1mu}
\frac{3x^8-2\rho^2x^6-6\sigma^6x^2+5\sigma^6\rho^2}{x^{11}
(-\rho+x)^2(x+\rho)^2}=0\qquad x=r_0.
$$

The graph of~$\Phi(r_0)$ for the given impact parameter
$\rho=2\sigma$ is shown in Fig.~\ref{fig5}.

\begin{figure}
\includegraphics{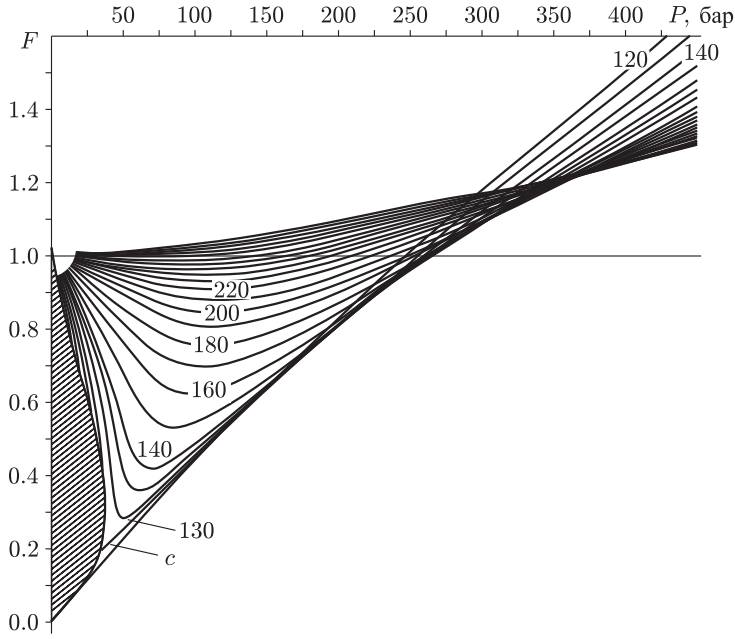}
\caption{The Hougen-Watson diagram for nitrogen; $c$ is the
critical temperature.} \label{fig5}
\end{figure}

For $E$ less than some value of~$E_0$, there appears a barrier
whose the depth is $E_{\max}^{\text{loc}}-E_{\min}^{\text{loc}}$
and for which the probabilities of penetration of particles into
the well~\cite{31:v851},~\cite{32:v851} due to thermal noise at a
given temperature are well known. This implies that, below this
barrier, the distributions that were obtained earlier for an
ideal gas are false, because there exists a probability of
penetration through the barrier at a given temperature. Thus, the
maximum of the barrier is the natural minimum for the energy
$(p_1-p_2)^2/m$, below which we cannot use the distribution given
above.

In~\cite{31:v851},~\cite{32:v851}, it was calculated how much
time a particle stays in a well of height~$h$ and depth~$\delta$
at a given temperature (thermal noise). If the number of
particles $N$ tends to infinity, then this time is proportional
to the number of particles occupying the well. It is a strange
fact that this is in agreement with our estimates~\cite{8:v851}
of penetration into the condensate and yields the value
$E_{\min}(\rho)$ for the condensate; for a given impact
parameter~$\rho$, this value corresponds to the same turning
point $r_0\cong1{.}21$ for the Lennard-Jones potential
with~$n=12$, to $r_0=1{.}02$ with~$n=18$, and to $r_0=1{.}28$
with~$n=7$.

But since~$\rho^3$ for an ideal gas is of the order of
$2V_{\text{sp}}\gg\rho^3$, then, accordingly, we can also define
$E_{\min}$ from Eq.~\eqref{eq63:v851} for prescribed values
of~$r_0$ and $\rho=\sqrt[3]{2V_{\text{sp}}}$; $E_{\min}$
corresponds to the gases for which the interaction between
particles is best described by the Lennard-Jones potential with a
given~$n$. If this is known, then $E_{\min}$ can be determined in
terms of~$V_{\text{sp}}$.

Using $\cP_\rho((p_1-p_2)^2/m)$, we define the one-dimensional
distribution of the difference of the momenta in the scattering
problem~\cite{11:v851},~\cite{33:v851}
\begin{equation}
\label{eq64:v851}  \cP_\rho\biggl(\frac{(p_1-p_2)^2}{m}\biggr)
=\frac{1}{e^{(p_1-p_2)^2/mkT-\mu'/kT}-1}
\biggl(\int_{p_{\min}}^\infty
\frac{d\xi}{e^{\xi^2/mkT-\mu'/kT}-1}\biggr)^{-1}.
\end{equation}

Here $\rho_{\min}=\sqrt{mE_{\min}}$, and, as pointed out above,
$E_{\min}$ is defined for the scattering problem by an
interaction in the form of the Lennard-Jones potential.

The distribution~\eqref{eq55:v851} contains the chemical
potential~$\mu_1$ which is related to the chemical
potential~$\mu$ for the distribution with a ``dressed'' potential
by relation~\eqref{eq65:v851} below, which expresses the fact
that, for a fixed scattering parameter~$\rho$, the number of
particles in the one-dimensional scattering problem is of the
order of $\sqrt[3]{N/2}$, where $N$ is the total number of
particles \emph{outside the condensate}.

The dressed potential~$u$ depends on the three-dimensional
momentum~$p$ and is independent of the coordinates under the
reduction to the scattering problem~\cite{13:v851}. Therefore,
the chemical potential~$\mu_1$ is connected to the chemical
potential $\mu=\mu(\rho)$ of the problem on the distribution with
a dressed potential by the relation
\begin{equation}
\label{eq65:v851} \frac23\mspace{1mu}\pi \int\frac{p^2\,dp}
{e^{(p^2/2m+u(p,\rho)-\mu)/kT}-1} =\biggl(\int_{p_{\min}}^\infty
\frac{dp}{e^{p^2/mkT-\mu_1/kT}-1}\biggr)^3.
\end{equation}

If $u(p)$ is positive, then $p_{\min}=0$, and hence $\mu_1<0$.

In the integral equation of the mean field~\eqref{eq55:v851}, we
can drop the external potential~$u_0$, because it is zero inside
the volume, and finally obtain~\cite{13:v851}
\begin{align}
\nonumber
u(p,\rho)+C(\mu,\rho)&=\frac89\mspace{1mu}\pi^2\mspace{1mu} \frac
NV\int_0^{\infty}\biggl[\int^{\infty}_{r_0(1/E,\rho)}
\frac{\Phi(r)r^2\,dr}{\sqrt{1-\rho^2/r^2-\Phi(r)/E}}
\\*
\nonumber &\hphantom{{}=\frac89\mspace{1mu}\pi^2\mspace{1mu}
\frac NV\int_0^{\infty}\biggl[}\quad\quad
\times\biggl[\int^{\infty}_{r_0}
\biggl\{\biggl(\sqrt{1-\frac{\rho^2}{r^2}
-\frac{\Phi(r)}E}\biggr)^{-1}-1\biggr\}\,dr\biggr]^{-1}
\\
\nonumber &\hphantom{{}=\frac89\mspace{1mu}\pi^2\mspace{1mu}
\frac NV\int_0^{\infty}\biggl[}\quad\quad
\times\frac{1}{e^{((p')^2/2m+u(p',\rho)-\mu)/kT}-1}
\\
\nonumber &\hphantom{{}=\frac89\mspace{1mu}\pi^2\mspace{1mu}
\frac NV\int_0^{\infty}\biggl[}\quad\quad \times
\frac{1}{e^{((p-p')^2/m-\mu_1)/kT}-1}(p')^2\,dp'\biggr]
\\&\qquad\qquad \times\biggl\{\int^\infty_{p_{\min}}
\frac{dp}{e^{(p^2/m-\mu_1)/kT}-1}\biggr\}^{-4}, \label{eq66:v851}
\end{align}
where $V$ is the volume, $r_0=r_0(\rho,E)$, and $E=(p-p')^2/m$.

As presented in~\cite{13:v851},
\begin{equation}
\label{eq67:v851}
 Z=\frac{2}{3kTV_{\text{sp}}^{2/3}}
\int_0^{V_{\text{sp}}^{1/3}}\rho\,d\rho\,
\int\biggl\{\frac{p^2}{2m}+u(p,
\rho)\biggr\}\cdot\frac{p^2}{2m}\frac{p^2dp}{e^{p^2/2m+u(p,\rho)-\mu/kT}-1}
\cdot\biggl\{\int\frac{p^2dp}{e^{p^2/2m+u(p,\rho)-\mu/kT}-1}
\biggr\}^{-1}.
\end{equation}

When the scattering problem is considered in the whole space,
then the distribution over the scattering section is uniform. But
we restrict the problem by the volume~$V_{\text{sp}}$. Then there
is no uniformity due to the boundary, at least, outside the
domain, where
$$
\frac{dZ}{dV_{\text{sp}}}<\frac Z{V_{\text{sp}}}\,.
$$
Under different assumptions, \footnote{For example,
A.~M.~Chebotarev proposed the following distribution of the
impact parameter: $\mathsf P(\rho\le r) = (1-r^2)^{3/2}$} the
distribution~$\mathcal P(\rho)$ over~$\rho$ can be different
(see~\cite{34:v851} (the Bertrand paradox),
\cite{35:v851},~\cite{36:v851}). Moreover, $Z$ is described by an
expression of type~\eqref{eq67:v851}, averaged with respect to
the distribution $\mathcal P(\rho)$ of the lines $\rho$ apart in
the ball of radius $\sqrt[3]{V_{\text{sp}}}$.

\begin{equation}
\label{eq67a:v851}
 Z=\frac{2}{3kTV_{\text{sp}}^{2/3}}
\int_0^{V_{\text{sp}}^{1/3}}{\mathcal P(\rho)}\rho\,d\rho\,
\int\biggl\{\frac{p^2}{2m}+u(p, \rho)\biggr\}\cdot
\frac{p^2dp}{e^{p^2/2m+u(p,\rho)-\mu/kT}-1}
\cdot\biggl\{\int\frac{p^2dp}{e^{p^2/2m+u(p,\rho)-\mu/kT}-1}
\biggr\}^{-1},
\end{equation}
where $\mu=\mu(\rho)$. Therefore, $Z$ can be taken at some mean
point $\rho_{\text{mean}}(T,V_{\text{sp}})$. Then
$$
\text{ $\frac{dZ}{dV_{\text{sp}}}
=\frac{dZ}{d\rho_{\text{mean}}}\cdot
\frac{d\rho_{\text{mean}}}{dV_{\text{sp}}} =-\infty$, }
$$
and hence,
$$
\text{ $\frac{dZ}{d\rho_{\text{mean}}}=-\infty$. }
$$

If, at this point, the asymptotics as $p\to 0$ of
$u(p,\rho_{\text{mean}})$ is of the form
$-p^2/2m+\alpha(\rho_{\text{mean}}) |\ln p|$, where
$d\alpha/d\rho_{\text{mean}}>0$, then this leads to the domain in
which $dZ/d\rho_{\text{mean}}=-\infty$, defining the
$\lambda$-transition to the condensate state and to the law
$$
\frac{dZ}{dP}=\frac{V_{\text{sp}}}{kT}
$$
for $P>P_\lambda$ and $V_{\text{sp}}>V_\lambda$ ($P_\lambda$ and
$V_\lambda$ depend on~$T$).

The equation for the dressed potential is of the form (we have
omitted the chemical potential for simplicity)
\begin{equation}
\label{eq68:v851} u(p,\rho)=\int_0^\infty
\frac{F((p-\eta)^2)\Theta((p-\eta)^2)\eta^2\,d\eta}
{e^{(p^2/2m+u(\eta,\rho))/kT}-1}
\cdot\biggl\{\int^\infty_{p_{\min}}
\frac{dp}{e^{p^2/mkT}-1}\biggr\}^{-4}-C,
\end{equation}
where \baselineskip 12.1pt
\begin{align}
\nonumber
 F(mE)&=\frac89\frac{\pi^2}{\rho^3}
\int^{\infty}_{r_0(1/E,\rho)}
\frac{\Phi(r)r^2\,dr}{\sqrt{1-\rho^2/r^2-\Phi(r)/E}}
\\&\qquad\qquad \times
\biggl[\int^{\infty}_{r_0}
\biggl\{\biggl(\sqrt{1-\frac{\rho^2}{r^2}-\frac{\Phi(r)}E}\biggr)^{-1}
-1\biggr\}\,dr\biggr]^{-1} \cdot\frac{1}{e^{E/kT}-1}\,,
\label{eq69:v851}
\end{align}
$\Theta((p-\eta)^2)$ is nonzero only in the domain $|p-p'|^2/m\ge
E_{\min}$,  $\rho^3=V_{\text{sp}}$, and $C=C(\mu,\rho)$.

Let us rewrite this equation in the form
\begin{equation}
\label{eq70:v851} u(p,\rho)=\frac12\int_{-\infty}^\infty
\frac{F((p-\eta)^2)\Theta((p-\eta)^2)\eta^2\,d\eta}
{e^{(\eta^2/2m+u(\eta,\rho))/kT}-1}\cdot
\biggl\{\int^\infty_{p_{\min}}
\frac{dp}{e^{p^2/mkT}-1}\biggr\}^{-4}-C
\end{equation}
and make the replacement $(p-\eta)^2=\xi^2$.

Then
\begin{align}
\nonumber u(p,\rho)&=\frac12\biggl\{
\int_{\sqrt{mE_{\min}}}^\infty \frac{F(\xi^2)\,d\xi}
{e^{(\eta^2/2m+u(\eta,\rho))/kT}-1}
+\int^{\sqrt{mE_{\min}}}_{-\infty} \frac{F(\xi^2)\,d\xi}
{e^{(\eta^2/2m+u(\eta,\rho))/kT}-1}\biggr\}
\\&\qquad\qquad \times\biggl\{\int^\infty_{p_{\min}}
\frac{dp}{e^{p^2/mkT}-1}\biggr\}^{-4}-C. \label{eq71:v851}
\end{align}
We search for conditions under which the solution of this
equation, as $p\to\infty$, is of the form
$$
u(p,\rho)=-\frac{p^2}{2m}+c(\rho),
$$
where $dc/d\rho>0$.

First, note that, by virtue of proofs and estimates similar to
those given in the theorems, we put the upper limit of the
integral over~$\eta$ equal to infinity, because $E$ in
\eqref{eq10:v851},~\eqref{eq12:v851},~\eqref{eq14:v851} (which is
different from $E=(p-p')^2/m$ in the scattering problem) is large,
while the integrand is rapidly decaying, and the difference
between the limit~$\sqrt{2mE}$ and~$\infty$ is less than the
given estimates.

But since we are concerned with the asymptotics of the solution
$u(p,\rho)$ as $p\to \infty$, it follows that, as $p'\to\infty$,
the integral over~$p'$ must sufficiently rapidly converge.
Therefore, this remark must be taken into account only for some
exotic family of solutions.

After the replacement indicated above, we express the
term~$\eta^2$ as $\eta^2=(\xi-p)^2$. But since the function is
symmetric, it follows that the integration of~$2p\xi$ over~$\xi$
yields zero. Thus, we find that the term on the right-hand side of
Eq.~\eqref{eq68:v851} is proportional to~$p^2$. Now it suffices
to equate the integral over~$\xi$ as $p\to\infty$ to~$-p^2/2m$.
Moreover, the choice of the constant~$c(\rho_{\text{mean}})$
remains arbitrary.

After the replacement, we obtain
$$
u(p,\rho)=-\frac{p^2}{2m}+w(p,\rho);
$$
here, as $p\to\infty$, we have $w(p,\rho)=c(\rho)|\ln p|$, and we
can write an equation for the function $w(p,\rho)> 0$,
$dc(\rho_{\text{mean}})/d\rho_{\text{mean}}>0$. Moreover, the
phase $\lambda$-transition, just as the minimal point of the
condensate, depends on the power of the repulsive term in the
Lennard-Jones potential as well on the quotients
$\gamma=\sigma/V_{\text{sp}}^{1/3}$ and $\alpha=U_0/kT$.
 The points~$\gamma_{\text{crit}}$
and the minimal point~$\alpha$, corresponding to
$dZ/dV_{\text{sp}}=-\infty$ are called \emph{$\lambda$-critical}.
As the pressure $w(p,\rho)$ increases above the point of the
$\lambda$-transition, $V_{\text{sp}}$ remains unchanged. This
implies that the volume~$V$ decreases as the pressure increases,
but, simultaneously, the number of particles outside the Bose
condensate also decreases. It is possible that all the particles
became dimers. and hence the total number of particles has
decreased. Further, they all became trimers, etc. The volume~$V$
has decreased, while the specific volume~$V_{\text{sp}}$ remained
constant---this is the law of the Bose condensate for classical
gases or, more precisely, is the law of cluster formation.

The energy of the   $\lambda$-point of  the logarithmical form
appears as $T=T_{\text{cr}}$. As one can see in Fig.~\ref{fig5}
and by virtue of
$$
\frac{dp}{dV_{\text{sp}}}|_{T=T_{\text{cr}}, p=p_{\text{cr}}} =
\frac{d^2p}{dV^2_{\text{sp}}}|_{T=T_{\text{cr}},
p=p_{\text{cr}}}=0
$$
the coefficient of incompressibility $\kappa= -(\pa\ln
V_{\text{sp}})/\pa p$  turns to infinity, and the compressibility
factor decreases steeply. This arouses a wave of  compressibility
(a shock wave), and therefore an additional term   $c|p|$, where
$c$ is the speed of sound, must occur in energy.

This term for $u(p, \rho)$ slows down the decrease of  $Z$, and
equation~ \eqref{eq66:v851} eliminates the shock wave as well as
this term.  At that instant for the derivative of heat capacity
the so-called $\lambda$-point occurs. This effect is shown more
obviously in Fig.~\ref{fig3}.

This phenomenon, as well as consequences of the
Pontryagin--An\-dro\-nov--Vitt theorem, does not follow from
classical mechanics but occurs when noise and fluctuations are
taken into account.  Therefore it does not follow from formulas
for the dressed potential, although the equations of ``collective
oscillations,'' as well as ``equations of variations,'' are
related to the dressed potential
\cite{Quasi-Particles_1,Quasi-Particles_2}.

To illustrate this we use both wave and quantum equations. The
wave equation of sound propagation has the form
$$
\frac {\pa^2 \Psi}{\pa t^2}= c^2\Delta\Psi,
$$
whereas the Schr{\"o}dinger equation is
\begin{equation}\label{ad3}
ih \frac{\pa \Psi}{\pa t} =\left\{- \frac{h^2}{2m}\Delta+
u(x)\right\}\Psi;
\end{equation}
or in the iterated form
\begin{equation}\label{ad4}
-h^2\frac{\pa^2\Psi}{\pa t^2} =\left\{ -\frac{h^2}{2m}\Delta+
u(x)\right\}^2\Psi.
\end{equation}

As follows from  formula (25) \cite{Quasi-Particles_2}
and~\cite{Quasi-Particles_3},
\begin{equation}\label{ad5}
h^2\frac{\pa^2\Psi}{\pa t^2} = h^2c^2\Delta\Psi-\left\{
-\frac{h^2}{2m}\Delta+ u(x)\right\}^2\Psi +\wh{O}(h^2)\Psi,
\end{equation}
where $\wh O(h^2)$ is an operator such that $\wh O e^{S(x, t)/h}
= O(h^2)$ for any $C^\infty$-smooth $S(x, t)$.

For $u(x) = 0$ and $\Psi = e^{(i/h)(px - Et)}$ this implies
\begin{equation}\label{ad6}
E^2= c^2p^2+\frac{p^4}{4m^2},
\end{equation}
which coincides with the spectrum obtained by N.N.~Bogolyubov for
the weakly nonideal classical gas~\cite{Bogol_47}.  This author
has established in 1995 that this spectrum has quasiclassical
rather than quantum nature
\cite{Quasi-Particles_1,Quasi-Particles_2} and, based on a
dependence from the capillary radius derived in~\cite{TMF_2005},
applied these results to the classical gas in
nanotubes~\cite{TMF_2007}.  These predictions are justified by
authoritative experimental data~\cite{Nature}.

The phase-space boxes~\eqref{ad2} are chosen to be invariant with
respect to the Hamiltonian system corresponding to the Hamiltonian
function $H(p, q)$.  In the context of quantum chaos
\cite{Quantization, Mishenko}, this corresponds to a kind of
generalized ergodicity.  Moreover, it follows from
\cite{Quasi-Particles_1,Quasi-Particles_2,Quasi-Particles_3,
Arxiv_Superfluidity} that the Hamiltonian
\begin{equation}
\label{Ham} \sqrt{c^2p^2 + H^2(p, q)}
\end{equation}
 corresponds to the Arnol'd diffusion in a
self-consistent field.

In this case, the equation for the dressed potential has the form
\begin{align}
\nonumber
u(p,\rho)+C(\mu,\rho)&=\frac43\mspace{1mu}\pi\mspace{1mu} \frac
{1}{V_{\text{sp}}}\int_0^{\infty}\biggl[\int^{\infty}_{r_0(E,\rho)}
\frac{\Phi(r)r^2\,dr}{\sqrt{1-\rho^2/r^2-\Phi(r)/E}}
\\*
\nonumber &\hphantom{{}=\frac89\mspace{1mu}\pi^2\mspace{1mu}
\frac NV\int_0^{\infty}\biggl[}\quad\quad
\times\biggl[\int^{\infty}_{r_0(E,\rho)}
\biggl\{\biggl(\sqrt{1-\frac{\rho^2}{r^2}
-\frac{\Phi(r)}E}\biggr)^{-1}-1\biggr\}\,dr\biggr]^{-1}
\\
\nonumber &\hphantom{{}=\frac89\mspace{1mu}\pi^2\mspace{1mu}
\frac NV\int_0^{\infty}\biggl[}\quad\quad
\times\frac{1}{\exp((\sqrt{cp^{'2}+p^{'4}/4m^2}+u(p',\rho))/kT)-1}
\\
\nonumber &\hphantom{{}=\frac89\mspace{1mu}\pi^2\mspace{1mu}
\frac NV\int_0^{\infty}\biggl[}\quad\quad \times
\frac{1}{\exp(((p-p')^2/m-\mu_1)/kT)-1}(p')^2\,dp'\biggr]
\\&\qquad\qquad \times\biggl\{\int^\infty_{p_{\min}}
\frac{dp}{\exp((p^2/m)/kT)-1}\biggr\}^{-3}
\\&\qquad\qquad \times\biggl(\int\frac{p^2\,dp}
{\exp((\sqrt{cp^{'2}+p^{'4}/4m^2}+u(p,\rho))/kT)-1}\biggr)^{-1},
\label{eq66:v851}
\end{align}
where $\mu$ and $\rho$ are given,
$C(\mu,\rho)$ is a constant depending on~$\mu$ and $\rho$,
$V_{\text{sp}}$ is the specific volume,
$r_0=r_0(\rho,E)$,
and $p_{\min}$ is determined by condition \eqref{eq65:v851}
for $\mu_1=\mu=0$.

After the change of variable $w(p, \rho) = u(p,
\rho_{\text{mean}}) - c|p|$ we find the values of $T$ and
$V_{\text{sp}}$ for which $w(p, \rho_{\text{mean}}) \approx
O(p^3)$ as $\mu(\rho_{\text{mean}})\to 0$.

It follows from the above theorems that as $N\to\infty$,
it is necessary to introduce a parameter $\kappa$
in the exponential in \eqref{eq65:v851}
in the left-hand side and a parameter $\kappa_1$
in the exponential in the right-hand side.
Since $p_{\min} \to 0$, we have $\kappa_1\gg\kappa$,
and hence the kernel of the integral operator tends to
the $\delta$-functoin as $\kappa \to 0$:
$\delta(p-p'+p_{\min})$.
To cancel the term $\frac{|p|}{c}$ in the right-hand side as $p\to0$,
it is necessary to satisfy the relation between $p_{\min}$ and
$\rho_{\text{mean}}$.

Then the leading term of the $T$ derivative of $TZ$, which
contains the heat capacity~$C_v$, will feature a logarithmic
dependence characteristic for a $\lambda$~point.  Indeed, as
$\mu(\rho_{\text{mean}}) \to 0$ we have
\begin{equation}
  \int \frac{O(p^3, \rho_{\text{mean}}) p^2\, dp}
  {(e^{O(p^3, \rho_{\text{mean}})} - 1)^2} \sim
  \int \frac{dp}p.
\end{equation}

Observe that although the quantum equations for the
self-consistent field go over into the classical ones as $h\to
0$, the equations of variations for the quantum mechanical
equations of the self-consistent field assume in the same limit
an extra term with respect to the classical equations of
variations (the equations of collective oscillations, see
\cite[1.1]{TMF_2007}).  This gives rise to the
Hamiltonian~\eqref{Ham}.

One can assume that at temperatures below the $\lambda$ point
ergodicity turns over into a KAM situation, making supefluidity
of a classical gas possible in a very thin nanotube capillary.
Thus the temperature of the $\lambda$ point can be regarded as
the crossover point between the generalized ergodicity and the
KAM dynamics.

\end{document}